%% file: main.tex
\documentclass[submission,copyright,creativecommons,conference]{eptcs}

\usepackage{etoolbox}

\newbool{arxiv}

\ifbool{arxiv}{}{\providecommand{\event}{QPL 2014}}
\usepackage{hyperref}
\usepackage[sort&compress,numbers,comma]{natbib}

\input{header}

\allowdisplaybreaks

\usepackage{wasysym}

\tikzset{
    downtriangle/.style={
        draw,
        shape border rotate=180,
        regular polygon,
        regular polygon sides=3,
        fill=black,
        node distance=1cm,
        minimum height=1.5em,
        inner sep=0pt
    }
}
\tikzset{
    uptriangle/.style={
        draw,
        shape border rotate=0,
        regular polygon,
        regular polygon sides=3,
        fill=black,
        node distance=1cm,
        minimum height=1.5em,
        inner sep=0pt
    }
}
\newcommand\smallwhitediagram{\ensuremath{\smash{\begin{aligned}
\begin{tikzpicture}[string, scale=0.2]
\node (black) [uptriangle, fill=white, minimum height=5pt] at (0,1) {};
\node (white) [downtriangle, fill=white, minimum height=5pt] at (1.1,0.2) {};
\draw [shorten <=-1pt, shorten >=-1pt] (black.corner 3) to [out=-30, in=150] (white.corner 3);
\draw [shorten >=-1pt] (-0.8,-0.5) to [out=up, in=-150] (black.corner 2);
\draw [shorten <=-1pt] (white.corner 1) to (1.1,-0.5);
\draw [shorten >=-1pt] (1.9,1.7) to [out=down, in=30] (white.corner 2);
\draw [shorten <=-1pt] (black.corner 1) to (0,1.7);
\end{tikzpicture}
\end{aligned}}}}
\newcommand\smallwhitediagramflip{\ensuremath{\smash{\begin{aligned}
\begin{tikzpicture}[string, scale=0.2]
\node (black) [uptriangle, fill=white, minimum height=5pt] at (0,1) {};
\node (white) [downtriangle, fill=white, minimum height=5pt] at (-1.1,0.2) {};
\draw [shorten <=-1pt, shorten >=-1pt] (black.corner 2) to [out=-150, in=30] (white.corner 2);
\draw [shorten >=-1pt] (0.8,-0.5) to [out=up, in=-30] (black.corner 3);
\draw [shorten <=-1pt] (white.corner 1) to (-1.1,-0.5);
\draw [shorten >=-1pt] (-1.9,1.7) to [out=down, in=150] (white.corner 3);
\draw [shorten <=-1pt] (black.corner 1) to (0,1.7);
\end{tikzpicture}
\end{aligned}}}}
\newcommand\smallblackdiagram{\ensuremath{\smash{\begin{aligned}
\begin{tikzpicture}[string, scale=0.2]
\node (black) [uptriangle, minimum height=5pt] at (0,1) {};
\node (white) [downtriangle, minimum height=5pt] at (1.1,0.2) {};
\draw [shorten <=-1pt, shorten >=-1pt] (black.corner 3) to [out=-30, in=150] (white.corner 3);
\draw [shorten >=-1pt] (-0.8,-0.5) to [out=up, in=-150] (black.corner 2);
\draw [shorten <=-1pt] (white.corner 1) to (1.1,-0.5);
\draw [shorten >=-1pt] (1.9,1.7) to [out=down, in=30] (white.corner 2);
\draw [shorten <=-1pt] (black.corner 1) to (0,1.7);
\end{tikzpicture}
\end{aligned}}}}
\newcommand\smallblackdiagramflip{\ensuremath{\smash{\begin{aligned}
\begin{tikzpicture}[string, scale=0.2]
\node (black) [uptriangle, minimum height=5pt] at (0,1) {};
\node (white) [downtriangle, minimum height=5pt] at (-1.1,0.2) {};
\draw [shorten <=-1pt, shorten >=-1pt] (black.corner 2) to [out=-150, in=30] (white.corner 2);
\draw [shorten >=-1pt] (0.8,-0.5) to [out=up, in=-30] (black.corner 3);
\draw [shorten <=-1pt] (white.corner 1) to (-1.1,-0.5);
\draw [shorten >=-1pt] (-1.9,1.7) to [out=down, in=150] (white.corner 3);
\draw [shorten <=-1pt] (black.corner 1) to (0,1.7);
\end{tikzpicture}
\end{aligned}}}}

\newcommand\whitecomonoid[1]{\ensuremath{(#1,\tinycomult[whitedot],\tinycounit[whitedot])}}
\newcommand\blackcomonoid[1]{\ensuremath{(#1,\tinycomult[blackdot],\tinycounit[blackdot])}}
\newcommand\graycomonoid[1]{\ensuremath{(#1,\tinycomult[graydot],\tinycounit[graydot])}}

\title{Abstract structure of unitary oracles for quantum algorithms}

\ifbool{arxiv}{
\date{5 June 2014}
    \author{
    \begin{tabular}{c@{\hspace{2cm}}c}
    W. J. Zeng\footnote{Department of Computer Science, University of Oxford} & Jamie Vicary${}^*$\footnote{Centre for Quantum Technologies, National University of Singapore}
    \\
    william.zeng@cs.ox.ac.uk & jamie.vicary@cs.ox.ac.uk
    \end{tabular}
    }
}{
    \author{William Zeng
    \institute{Department of Computer Science, University of Oxford}
    \email{william.zeng@cs.ox.ac.uk}
    \and
    Jamie Vicary
    \institute{Centre for Quantum Technologies, University of Singapore\\
    and Department of Computer Science, University of Oxford}
    \email{jamie.vicary@cs.ox.ac.uk}
    }
\def\titlerunning{Abstract structure of unitary oracles for quantum algorithms}
\def\authorrunning{W. J. Zeng \& J. Vicary}
}

\begin{document}

\maketitle

\begin{abstract}
We show that a pair of complementary dagger-Frobenius algebras, equipped with a self-conjugate comonoid homomorphism onto one of the algebras, produce a nontrivial unitary morphism on the product of the algebras. This gives an abstract understanding of the structure of an oracle in a quantum computation, and we apply this understanding to develop a new algorithm for  the deterministic identification of group homomorphisms into abelian groups. We also discuss an application to the categorical theory of signal-flow networks.
\end{abstract}

\section{Introduction}

\subsection{Overview}

Pairs of complementary dagger-Frobenius algebras play an important role in the high-level characterization of quantum phenomena~\cite{vicary-tqa, bobross}, as the algebraic content of mutually unbiased bases. In Section~\ref{sec:unitary}, we show that if a such a pair is equipped with a self-conjugate comonoid homomorphism onto one of the algebras, a \textit{unitary} map can be constructed that has the same abstract structure as an \textit{oracle} in the theory of quantum algorithms. This gives insight into the logical structure of quantum algorithms and opens up a new avenue for their generalization.

Most known quantum algorithms are constructed using these black-box quantum oracles, whose structure can be depicted graphically in the following way:
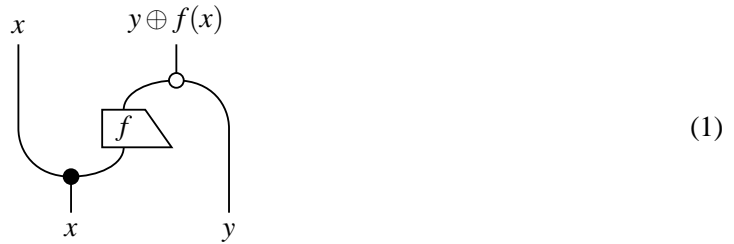
\begin{equation}
\label{eq:qoracle}
\begin{aligned}
\begin{tikzpicture}[string,yscale=\newyscale,yscale=0.8]
    \node (dot) [blackdot] at (0,1) {};
    \node (f) [morphism, wedge] at (0.7,2) {$f$};
    \node (m) [whitedot] at (1.4,3) {};
\draw (0,0.25)
        node [below] {$x$}
    to (0,1)
    to [out=left, in=south] (-0.7,2)
    to (-0.7,3.75)
        node [above] {$x$};
\draw (0,1)
    to [out=right, in=south] (f.south);
\draw  (f.north)
    to [out=up, in=left] (1.4,3)
    to [out=right, in=up] +(0.7,-1)
    to (2.1,0.25)
        node [below] {$y$};;
\draw (m.center) to +(0,0.75) node [above] {$y \oplus f(x)$};
\end{tikzpicture}
\end{aligned}
\end{equation}
 Here we read the diagram from bottom to top, defining a map of type $\mathbb{C}^n\otimes\mathbb{C}^m \to \mathbb{C}^n\otimes\mathbb{C}^m$ that acts as $|x\rangle \otimes |y \rangle\mapsto|x\rangle \otimes |y\oplus f(x)\rangle$ for a group product $\oplus$. Section~\ref{sec:unitary} contains a full abstract description. Oracle-based algorithms include the Deutsch-Jozsa, Grover, and hidden subgroup algorithms.  In the Deutsch-Jozsa and Grover algorithms the oracle implements a function \mbox{$f:S\to\{0,1\}$} where $S$ is a finite set. In the hidden subgroup algorithm, the oracle implements a function $f:G \to S$ where $G$ is a finite group and $S$ is a finite set. In \cite{vicary-tqa} it was shown that the unitary oracle described in Section \ref{sec:unitary} characterizes the structure of these well-known algorithms.
\def\licsscale{0.70}

For these oracles to be physically implementable, they must be \textit{unitary operators}. In this paper we give an abstract proof of unitarity for these operators using categorical algebra. In Section~\ref{sec:algorithm} we apply this result to develop a new quantum algorithm for the identification of group homomorphisms into an abelian group, in a number of queries which is equal to the number of simple factors of the target group. The graphical approach provides a simple proof of correctness of the algorithm, and leads to an algorithm which is more general than existing work in the literature~\cite{hoyer-conjops}.

In Section~\ref{sec:signalflow} we investigate an application to the theory of signal-flow networks~\cite{baezerbele, fong-transfer, sobocinski}. We show that the formalism contains dagger-Frobenius algebras equipped with self-conjugate homomorphisms, and that, as a consequence, the network representing a single resistor is unitary.

\paragraph{Acknowledgements.} We are grateful to John Baez and Pawel Sobocinski for useful discussions about signal-flow networks. Section~4 of this paper has some technical overlap with~\cite{sobocinski} and was prepared independently. We are grateful to the authors for pointing out their work to us in the prepublication phase of this article. Will Zeng acknowledges the support of the Rhodes Trust in funding this work.

\subsection{Frobenius monoids and complementarity}

In this Section we collect some standard results from the literature~\cite{bobross}. We assume some familiarity with the graphical calculus for symmetric monoidal dagger-categories~\cite{selinger}. We use a notation in which morphisms are drawn from bottom-to-top.

\begin{defn}
In a monoidal category, a \textit{comonoid} is a triple \whitecomonoid{A} of an object $A$, a morphism $\tinycomult[whitedot] : A \to A \otimes A$ called the comultiplication, and a morphism $\tinycounit[whitedot] : A \to I$ called the counit, satisfying coassociativity and counitality equations:
\def\frobscale{0.5}
\begin{calign}
\begin{aligned}
\begin{tikzpicture}[thick, scale=0.7, yscale=-1]
\draw (0,0) to [out=up, in=\swangle] (0.5, 1);
\draw (1,0) to [out=up, in=\seangle] (0.5,1);
\draw (2,0) to [out=up, in=\seangle] (1.25,2);
\draw (0.5,1) to [out=up, in=\swangle] (1.25, 2);
\draw (1.25,2) to (1.25, 3);
\node [whitedot] at (0.5,1) {};
\node [whitedot] at (1.25,2) {};
\end{tikzpicture}
\end{aligned}
\quad=\quad
\begin{aligned}
\begin{tikzpicture}[xscale=-1, thick, scale=0.7, yscale=-1]
\draw (0,0) to [out=up, in=\swangle] (0.5, 1);
\draw (1,0) to [out=up, in=\seangle] (0.5,1);
\draw (2,0) to [out=up, in=\seangle] (1.25,2);
\draw (0.5,1) to [out=up, in=\swangle] (1.25, 2);
\draw (1.25,2) to (1.25, 3);
\node [whitedot] at (0.5,1) {};
\node [whitedot] at (1.25,2) {};
\end{tikzpicture}
\end{aligned}
&\qquad\qquad\qquad
\begin{aligned}
\begin{tikzpicture}[thick, scale=0.7, yscale=-1]
\draw (0,-1.5) to (0,-0.5) to [out=up, in=\swangle] (0.75,0.5) node [whitedot] {} to (0.75,1.5);
\draw (1.5,-0.5) node [whitedot] {} to [out=up, in=\seangle] (0.75,0.5);
\end{tikzpicture}
\end{aligned}
\quad=\quad
\begin{aligned}
\begin{tikzpicture}[thick, scale=0.7, yscale=-1]
\draw (0,0) to (0,3);
\end{tikzpicture}
\end{aligned}
\quad=\quad
\begin{aligned}
\begin{tikzpicture}[thick, xscale=-1, scale=0.7, yscale=-1]
\draw (0,-1.5) to (0,-0.5) to [out=up, in=\swangle] (0.75,0.5) node [whitedot] {} to (0.75,1.5);
\draw (1.5,-0.5) node [whitedot] {} to [out=up, in=\seangle] (0.75,0.5);
\end{tikzpicture}
\end{aligned}
\end{calign}
\end{defn}

\noindent
In a monoidal dagger-category, we can apply the dagger operation to these structures to obtain the associated monoid. We can then ask for the comonoid and monoid to interact in various ways.
\begin{defn}
In a monoidal dagger-category, a comonoid \whitecomonoid{A} is \emph{dagger-Frobenius} when the following equation holds:
\begin{equation}\label{eq:frobenius}
\begin{aligned}
\begin{tikzpicture}[scale=0.7, thick]
    \draw (0,0) to (0,1) to [out=up, in=\swangle] (0.5,2) node [whitedot] {} to (0.5,3);
    \draw (0.5,2) to [out=\seangle, in=\nwangle] (1.5,1) node [whitedot] {};
    \draw (1.5,0) to (1.5,1) to [out=\neangle, in=down] (2,2) to (2,3);
\end{tikzpicture}
\end{aligned}
    \quad = \quad
\begin{aligned}
\begin{tikzpicture}[scale=0.7, thick, xscale=-1]
    \draw (0,0) to (0,1) to [out=up, in=\swangle] (0.5,2) node [whitedot] {} to (0.5,3);
    \draw (0.5,2) to [out=\seangle, in=\nwangle] (1.5,1) node [whitedot] {};
    \draw (1.5,0) to (1.5,1) to [out=\neangle, in=down] (2,2) to (2,3);
\end{tikzpicture}
\end{aligned}
  \end{equation}
\end{defn}
\begin{defn}
In a symmetric monoidal dagger-category, a \textit{classical structure} is a commutative dagger-Frobenius comonoid \whitecomonoid{A} satisfying the \textit{specialness} condition:
\begin{equation}
\begin{aligned}
\begin{tikzpicture}[thick, scale=0.7]
\draw (0,0.25) to (0,1) node [whitedot] {} to [out=\nwangle, in=down] (-0.5,1.5) to [out=up, in=\swangle] (0,2) node [whitedot] {} to (0,2.75);
\draw (0,1) to [out=\neangle, in=down] (0.5,1.5) to [out=up, in=\seangle] (0,2);
\end{tikzpicture}
\end{aligned}
\quad=\quad
  \begin{aligned}
  \begin{tikzpicture}[yscale=\newyscale, thick, yscale=0.7]
  \draw (-0.5,0) to (-0.5,3.2);
  \end{tikzpicture}
  \end{aligned}
  \end{equation}
\end{defn}

\begin{defn}
In a symmetric monoidal dagger-category, a dagger-Frobenius comonoid is \emph{symmetric} when the following condition holds:
\begin{equation}
\begin{aligned}
\begin{tikzpicture}
\draw (0,0) node [whitedot] {} to (0,0.5) node [whitedot] {} to [out=\nwangle, in=down] (-0.5,1.0) to [out=up, in=down] (0.5,2);
\draw (0,0.5) to [out=\neangle, in=down] (0.5,1) to [out=up, in=down] (-0.5,2);
\end{tikzpicture}
\end{aligned}
\quad=\quad
\begin{aligned}
\begin{tikzpicture}
\draw (0,0) node [whitedot] {} to (0,0.5) node [whitedot] {} to [out=\nwangle, in=down] (-0.5,1.0) to [out=up, in=down] (-0.5,2);
\draw (0,0.5) to [out=\neangle, in=down] (0.5,1) to [out=up, in=down] (0.5,2);
\end{tikzpicture}
\end{aligned}
\end{equation}
\end{defn}

\begin{defn}
In a symmetric monoidal dagger-category, the \emph{dimension} $d(A)$ of an object $A$ equipped with a dagger-Frobenius comonoid \whitecomonoid{A} is given by the following composite:
\begin{equation}\label{eq:dim}
    \ud(A)
    \quad := \quad
    \begin{aligned}\begin{tikzpicture}[yscale=\newyscale,xscale=-1]
        \node (0) at (0,0) {};
        \node[whitedot] (1) at (0,0.66) {};
        \node (2) at (-0.5,1.2) {};
        \node (3) at (0.5,1.2) {};
        \node (4) at (-0.5,2.0) {};
        \node (5) at (0.5,2.0) {};
        \draw[string] (0.center) to (1.center);
        \draw[string, out=180, in=270] (1.center) to (2.center);
        \draw[string, out=0, in=270] (1.center) to (3.center);
        \draw[string, out=90, in=270] (2.center) to (5.center);
        \draw[string, out=90, in=270] (3.center) to (4.center);;
        \node[whitedot] (6) at (0,2.54) {};        
        \node (7) at (0,3.2) {};
        \draw[string] (0.center) node [whitedot] {} to (1);
        \draw[string] (6.center) to (7.center) node [whitedot] {$$};
        \draw[string, in=left, out=up] (4.center) to node [auto] {$$} (6.center);
        \draw[string, in=right, out=up] (5.center) to node [auto, swap] {$$} (6.center);
    \end{tikzpicture}\end{aligned}
  \end{equation}
\end{defn}

\noindent
When the algebra is commutative and special, equation~\eqref{eq:dim} can be simplified to the composition of the unit and counit.

\begin{defn}[Complementarity]
\label{def:complementarity}
In a symmetric monoidal dagger-category, two special symmetric dagger-Frobenius comonoids \whitecomonoid{A} and \graycomonoid{A} are \emph{complementary} when the following equation holds:
\begin{equation}
\label{eq:complementarity}
\ud(A) \,\,
\begin{pic}[string, yscale=\newyscale, xscale=\newxscale]
\draw (-0.5,0.25) to (-0.5,1) node [graydot] {} to [out=left, in=right] (-1,2) node [graydot] {} to [out=left, in=right] (-1.5,1.5) node [whitedot] {} to [out=left, in=down] (-2,2) to [out=up, in=left] (-0.75,3) node (a) [whitedot] {} to [out=right, in=right] (-0.5,1);
\draw (a.center) to +(0,0.75);
\end{pic}
\quad=\quad\,\,\,
\begin{pic}[string, yscale=\newyscale, xscale=\newxscale]
\draw (0,0.25) to (0,1) node [graydot] {};
\draw (0,3) node [whitedot] {} to (0,3.75);
\end{pic}
\end{equation}
\end{defn}

\noindent
Note that this is not a symmetric condition between the gray and white structures. However, thanks to the symmetric property of the dagger-Frobenius algebras, it is equivalent to the following alternative condition:
\begin{equation}
\ud(A) \,\,
\begin{pic}[string, yscale=\newyscale, xscale=\newxscale, yscale=-1]
\draw (-0.5,0.25) to (-0.5,1) node [whitedot] {} to [out=left, in=right] (-1,2) node [whitedot] {} to [out=left, in=right] (-1.5,1.5) node [graydot] {} to [out=left, in=down] (-2,2) to [out=up, in=left] (-0.75,3) node (a) [graydot] {} to [out=right, in=right] (-0.5,1);
\draw (a.center) to +(0,0.75);
\end{pic}
\quad=\quad\,\,\,
\begin{pic}[string, yscale=\newyscale, xscale=\newxscale]
\draw (0,0.25) to (0,1) node [graydot] {};
\draw (0,3) node [whitedot] {} to (0,3.75);
\end{pic}
\end{equation}
The daggers of these equations give rise to two further equivalent conditions.

By the symmetric property of the dagger-Frobenius algebras, this condition is equivalent to

\begin{defn}
In a monoidal dagger-category, a comonoid homomorphism \mbox{$f:\blackcomonoid{A} \to \graycomonoid{B}$} between dagger-Frobenius comonoids is \emph{self-conjugate} when the following property holds:
\begin{equation}
\label{eq:comonoidhomomorphismselfconjugate}
\begin{aligned}
\begin{tikzpicture}[yscale=\newyscale, xscale=\newxscale, thick]
\node [morphism, wedge] (f) at (2,1) {$f$};
\draw (0,-1) to [out=up, in=left, in looseness=0.9] (1,2) node [graydot] {} to (1,2.5) node [graydot] {};
\draw (1,2) to [out=right, in=up] (f.north);
\draw (f.south) to [out=down, in=left] (3,0) node [blackdot] {} to [out=right, in=down, out looseness=0.9] (4,3);
\draw (3,0) to (3,-0.5) node [blackdot] {};
\node [graydot] at (1,2) {};
\end{tikzpicture}
\end{aligned}
\quad=\quad
\begin{aligned}
\begin{tikzpicture}[string]
\node (f) at (0,0) [morphism, wedge, hflip] {$f$};
\draw (0,-1.5) to (f.south);
\draw (f.north) to (0,1.5);
\end{tikzpicture}
\end{aligned}
\end{equation}
\end{defn}
\begin{lemma}
\label{lem:comonoidhomomorphismselfconjugate}
In {\bf Hilb}, comonoid homomorphisms $f:\blackcomonoid{A} \to \graycomonoid{B}$ of classical structures are self-conjugate.
\end{lemma}
\begin{proof}
Recall that comonoid homomorphisms between classical structures in \cat{Hilb} are exactly classical functions between the copyable points~\cite{OrthBasis:2008}. The linear maps on either side of~\eqref{eq:comonoidhomomorphismselfconjugate} will be the same if and only if their matrix elements are the same, obtained by composing with $\ket i$ at the bottom and $\bra j$ at the top. On the left-hand side, this gives the following result:
\begin{equation}
\begin{aligned}
\begin{tikzpicture}[yscale=\newyscale, xscale=\newxscale, thick]
\node [morphism, wedge] (f) at (2,1) {$f$};
\draw (0,0) node [state] {$i$} to [out=up, in=left, in looseness=0.9] (1,2) node [graydot] {} to (1,2.5) node [graydot] {};
\draw (1,2) to [out=right, in=up] (f.north);
\draw (f.south)
    to [out=down, in=left] (3,0)
        node [blackdot] {}
    to [out=right, in=down, out looseness=0.9] (4,2)
        node [state, hflip] {$j$};
\draw (3,0) to (3,-0.5) node [blackdot] {};
\end{tikzpicture}
\end{aligned}
\quad=\quad
\begin{aligned}
\begin{tikzpicture}[string]
\node (f) [morphism, wedge] at (0,0) {$f$};
\draw (0,-0.75) node [state] {$j$} to (f.south);
\draw (0,0.75) node [state, hflip] {$i$} to (f.north);
\end{tikzpicture}
\end{aligned}
\quad=\quad
\left\{
\begin{array}{ll}
1 & \text{ if } i=f(j), \\
0 & \text{ if } i \neq f(j).
\end{array}
\right.
\end{equation}
On the right we can do this calculation:
\begin{equation}
\begin{aligned}
\begin{tikzpicture}[string, yscale=\newyscale, xscale=\newxscale]
\node (f) [morphism, wedge, hflip] at (0,0) {$f$};
\draw (0,-0.75) node [state] {$i$} to (f.south);
\draw (0,0.75) node [state, hflip] {$j$} to (f.north);
\end{tikzpicture}
\end{aligned}
\quad=\quad
\left(
\begin{aligned}
\begin{tikzpicture}[string]
\node (f) [morphism, wedge] at (0,0) {$f$};
\draw (0,-0.75) node [state] {$j$} to (f.south);
\draw (0,0.75) node [state, hflip] {$i$} to (f.north);
\end{tikzpicture}
\end{aligned}
\right) ^\dag
\quad=\quad
\left\{
\begin{array}{ll}
1 & \text{ if } i=f(j) \\
0 & \text{ if } i \neq f(j)
\end{array}
\right\}^\dag 
\quad = \quad
\left\{
\begin{array}{ll}
1 & \text{ if } i=f(j), \\
0 & \text{ if } i \neq f(j).
\end{array}
\right.
\end{equation}
This is the same result as for the left-hand side, and so expression~\eqref{eq:comonoidhomomorphismselfconjugate} holds.
\end{proof}

\section{Unitary oracles}

\label{sec:unitary}

\subsection{Complementarity via unitarity}
A pair of symmetric dagger-Frobenius algebras can be used to build a linear map in the following way:
\begin{equation}
\label{eq:generalizedcnot}
\sqrt{\ud(A)}\,\,
\begin{aligned}
\begin{tikzpicture}[yscale=\newyscale,string]
\node (b) [graydot] at (0,0) {};
\node (w) [whitedot] at (1,1) {};
\draw (-0.75,2) to [out=down, in=left] (b.center);
\draw (b.center) to [out=right, in=left] (w.center);
\draw (w.center) to (1,2);
\draw (b.center) to (0,-1);
\draw (w.center) to [out=right, in=up] (1.75,-1);
\end{tikzpicture}
\end{aligned}
\end{equation}
Here we have assumed that we operate in a category where square roots of scalars exist.  The two algebras are complementary exactly when this composite is unitary, as we show in the following theorem.

\begin{theorem}[Complementarity via a unitary]
\label{thm:complementarityunitary}
  In a dagger symmetric monoidal category, two symmetric dagger-Frobenius algebras are complementary if and only if the composite~\eqref{eq:generalizedcnot} is unitary.
\end{theorem}
\begin{proof}
  Composing ~\eqref{eq:generalizedcnot} with its adjoint in one order, we obtain the following:
  \tikzset{every picture/.style={scale=0.95,yscale=0.8}}
  \begin{equation}
  \label{eq:generalizedcnotunitaryproof}
  \begin{aligned}
  \begin{tikzpicture}[yscale=\newyscale, xscale=\newxscale,string]
  \node at (-1.5,-2.6) {$\ud (A)$};
  \node (A) at (0,0) {};
  \node (B) at (1.75,0) {};
  \node (b1) [graydot] at (0,-1) {};
  \node (w1) [whitedot] at (1,-2) {};
  \node (w2) [whitedot] at (1,-3) {};
  \node (b2) [graydot] at (0,-4) {};
  \node (C) at (0,-5) {};
  \node (D) at (1.75,-5) {};
  \draw (A.center) to (b1.center);
  \draw (b1.center) to [out=right, in=left] (w1.center);
  \draw (w1.center) to (w2.center);
  \draw (w2.center) to [out=left, in=right] (b2.center);
  \draw (b2.center) to (C.center);
  \draw (w2.center) to [out=right, in=up] (D.center);
  \draw (w1.center) to [out=right, in=down] (B.center);
  \draw (b1.center) to [out=left, in=left] (b2.center);
  \end{tikzpicture}
  \end{aligned}
  \,\,=
  \,\,
  \hspace{-3pt}
  \begin{aligned}
  \begin{tikzpicture}[yscale=\newyscale, xscale=\newxscale,string]
  \node at (-2.1,-2.6) {$\ud (A)$};
  \node (A) at (-1.75,0) {};
  \node (B) at (0.5,0) {};
  \node (w1) [whitedot] at (0.5,-1.0) {};
  \node (w2) [whitedot] at (1.25,-3) {};
  \node (w3) [whitedot] at (0,-2) {};
  \node (b1) [graydot] at (-1,-2) {};
  \node (b2) [graydot] at (0,-4) {};
  \node (b3) [graydot] at (-0.5,-1) {};
  \node (C) at (0,-5) {};
  \node (D) at (2,-5) {};
  \draw (A.center) to [out=down, in=left] (b1.center);
  \draw (w1.center) to [out=right, in=up] (w2.center);
  \draw (w2.center) to [out=left, in=right] (b2.center);
  \draw (b2.center) to (C.center);
  \draw (w2.center) to [out=right, in=up] (D.center);
  \draw (w1.center) to [out=up, in=down] (B.center);
  \draw (b1.center) to [out=down, in=left] (b2.center);
  \draw (b1.center) to [out=right, in=left] (b3.center);
  \draw (b3.center) to [out=right, in=left] (w3.center);
  \draw (w3.center) to [out=right, in=left] (w1.center);
  \end{tikzpicture}
  \end{aligned}
  \,\,= \,\,
  \begin{aligned}
  \begin{tikzpicture}[yscale=\newyscale, xscale=\newxscale,string,yscale=0.833]
  \node at (-1.6,-3.1) {$\ud (A)$};  
  \node (A) at (-1,0) {};
  \node (B) at (1.5,0) {};
  \node (w1) [whitedot] at (1.5,-1.0) {};
  \node (w2) [whitedot] at (1.0,-2) {};
  \node (w3) [whitedot] at (0.5,-3) {};
  \node (b1) [graydot] at (0.5,-4) {};
  \node (b2) [graydot] at (0,-5) {};
  \node (b3) [graydot] at (0,-2) {};
  \node (C) at (0,-6) {};
  \node (D) at (2.5,-6) {};
  \draw (A.center) to [out=down, in=left, in looseness=0.6] (b2.center);
  \draw (w1.center) to [out=left, in=up] (w2.center);
  \draw (w2.center) to [out=right, in=right] (b1.center);
  \draw (b2.center) to (C.center);
  \draw (w1.center) to [out=right, in=up, out looseness=0.6] (D.center);
  \draw (w1.center) to [out=up, in=down] (B.center);
  \draw (b1.center) to [out=down, in=right] (b2.center);
  \draw (b1.center) to [out=left, in=left] (b3.center);
  \draw (b3.center) to [out=right, in=left] (w3.center);
  \draw (w2.center) to [out=left, in=right] (w3.center);
  \end{tikzpicture}
  \end{aligned}
  \end{equation}
  If the complementarity condition~\eqref{eq:complementarity} holds then this is clearly the identity on \mbox{$A \otimes A$}. The other composite can be shown to be the identity in a similar way, and so~\eqref{eq:generalizedcnot} is unitary.

  Conversely, suppose~\eqref{eq:generalizedcnot} is unitary. Then the final expression of~\eqref{eq:generalizedcnotunitaryproof} certainly equals the identity on $A \otimes A$:
\begin{equation}
  \begin{aligned}
  \begin{tikzpicture}[yscale=\newyscale, xscale=\newxscale,string,yscale=0.833]
  \draw (0.5,-5) to (0.5,1);
  \draw (-1.5,-5) to (-1.5,1);
  \end{tikzpicture}
  \end{aligned}
\quad=\quad\hspace{-8pt}
\begin{aligned}
  \begin{tikzpicture}[yscale=\newyscale, xscale=\newxscale,string,yscale=0.833]
  \node at (-1.6,-3.1) {$\ud (A)$};  
  \node (A) at (-1,0) {};
  \node (B) at (1.5,0) {};
  \node (w1) [whitedot] at (1.5,-1.0) {};
  \node (w2) [whitedot] at (1.0,-2) {};
  \node (w3) [whitedot] at (0.5,-3) {};
  \node (b1) [graydot] at (0.5,-4) {};
  \node (b2) [graydot] at (0,-5) {};
  \node (b3) [graydot] at (0,-2) {};
  \node (C) at (0,-6) {};
  \node (D) at (2.5,-6) {};
  \draw (A.center) to [out=down, in=left, in looseness=0.6] (b2.center);
  \draw (w1.center) to [out=left, in=up] (w2.center);
  \draw (w2.center) to [out=right, in=right] (b1.center);
  \draw (b2.center) to (C.center);
  \draw (w1.center) to [out=right, in=up, out looseness=0.6] (D.center);
  \draw (w1.center) to [out=up, in=down] (B.center);
  \draw (b1.center) to [out=down, in=right] (b2.center);
  \draw (b1.center) to [out=left, in=left] (b3.center);
  \draw (b3.center) to [out=right, in=left] (w3.center);
  \draw (w2.center) to [out=left, in=right] (w3.center);
  \end{tikzpicture}
  \end{aligned}
\end{equation}
 Composing with the black counit at the top-left and the white unit at the bottom-right then gives back complementarity condition~\eqref{eq:complementarity} as required:
\begin{equation}
  \begin{aligned}
  \begin{tikzpicture}[yscale=\newyscale, xscale=\newxscale,string,yscale=0.833]
  \node (w) [whitedot] at (-0.5,-4) {};
  \node (b) [graydot] at (-1.5,0) {}; 
  \draw (-0.5,-4) to (-0.5,1);
  \draw (-1.5,-5) to (-1.5,0);
  \end{tikzpicture}
  \end{aligned}
  \,\,\,\,= 
  \begin{aligned}
  \begin{tikzpicture}[yscale=\newyscale, xscale=\newxscale,string,yscale=0.833]
  \node at (-1.6,-3.1) {$\ud (A)$}; 
  \node (w) [graydot] at (-1,0) {};
  \node (b) [whitedot] at (2.5,-6) {}; 
  \node (A) at (-1,0) {};
  \node (B) at (1.5,0) {};
  \node (w1) [whitedot] at (1.5,-1.0) {};
  \node (w2) [whitedot] at (1.0,-2) {};
  \node (w3) [whitedot] at (0.5,-3) {};
  \node (b1) [graydot] at (0.5,-4) {};
  \node (b2) [graydot] at (0,-5) {};
  \node (b3) [graydot] at (0,-2) {};
  \node (C) at (0,-6) {};
  \node (D) at (2.5,-6) {};
  \draw (A.center) to [out=down, in=left, in looseness=0.6] (b2.center);
  \draw (w1.center) to [out=left, in=up] (w2.center);
  \draw (w2.center) to [out=right, in=right] (b1.center);
  \draw (b2.center) to (C.center);
  \draw (w1.center) to [out=right, in=up, out looseness=0.6] (D.center);
  \draw (w1.center) to [out=up, in=down] (B.center);
  \draw (b1.center) to [out=down, in=right] (b2.center);
  \draw (b1.center) to [out=left, in=left] (b3.center);
  \draw (b3.center) to [out=right, in=left] (w3.center);
  \draw (w2.center) to [out=left, in=right] (w3.center);
  \end{tikzpicture}
  \end{aligned}
  \,\,= \,\,  \mathrm{d}(A)
  \begin{pic}[string, yscale=1, xscale=\newxscale]
\draw (-0.5,0.25) to (-0.5,1) node [graydot] {} to [out=left, in=right] (-1,2) node [graydot] {} to [out=left, in=right] (-1.5,1.5) node [whitedot] {} to [out=left, in=down] (-2,2) to [out=up, in=left] (-0.75,3) node (a) [whitedot] {} to [out=right, in=right] (-0.5,1);
\draw (a.center) to +(0,0.75);
\end{pic}
\end{equation}
This completes the proof.
\end{proof}

\subsection{Families of unitary oracles}

This pair of complementary observables automatically gives rise to a much larger family of unitaries, one for each self-conjugate comonoid homomorphism onto one of the classical structures in the pair. See equation~\eqref{eq:comonoidhomomorphismselfconjugate} for the definition of the self-conjugacy property. Lemma~\ref{lem:comonoidhomomorphismselfconjugate} demonstrated that in \cat{FHilb}, every comonoid homomorphism of classical structures is self-conjugate.
\begin{defn}[Oracle]
\label{oracle}
In a symmetric monoidal dagger-category, given a dagger-Frobenius comonoid $\blackcomonoid{A}$, a pair of complementary symmetric dagger-Frobenius comonoids \graycomonoid{B} and \whitecomonoid{B}, and a self-conjugate comonoid homomorphism $f : \blackcomonoid{A} \to \graycomonoid{B}$, the \emph{oracle} is defined to be the following endomorphism of $A \otimes B$:
\begin{equation}
\label{eq:oracle}
\sqrt{\ud(A)}
\begin{aligned}
\begin{tikzpicture}[string,yscale=\newyscale,yscale=0.8]
    \node (dot) [blackdot] at (0,1) {};
    \node (f) [morphism, wedge] at (0.7,2) {$f$};
    \node (m) [whitedot] at (1.4,3) {};
\draw (0,0.25)
        node [below] {$A$}
    to (0,1)
    to [out=left, in=south] (-0.7,2)
    to (-0.7,3.75)
        node [above] {$A$};
\draw (0,1)
    to [out=right, in=south] (f.south);
\draw  (f.north)
    to [out=up, in=left] (1.4,3)
    to [out=right, in=up] +(0.7,-1)
    to (2.1,0.25)
        node [below] {$B$};;
\draw (m.center) to +(0,0.75) node [above] {$B$};
\end{tikzpicture}
\end{aligned}
\end{equation}
\end{defn}
\begin{theorem}
\label{thm:familyofunitaries}
Oracles are unitary.
\end{theorem}
\begin{proof}
To demonstrate that the oracle~\eqref{eq:oracle} is unitary, we must compose it with its adjoint on both sides and show that we get the identity in each case. In one case, we obtain the following, making use of the Frobenius laws, self-conjugacy of $f$, associativity and coassociativity, the fact that $f$ preserves comultiplication, the complementarity condition, the fact that $f$ preserves the counit, and the unit and counit laws:
\tikzset{every picture/.style={scale=0.9,yscale=0.9}}
\begin{align*}
&\ud(A)
\begin{aligned}
\begin{tikzpicture}[yscale=0.6,string,xscale=1]
\node (A) at (0,2) {};
\node (B) at (1.75,0) {};
\node (b1) [blackdot] at (0,1) {};
\node (w1) [whitedot] at (1,-1) {};
\node (w2) [whitedot] at (1,-2) {};
\node (b2) [blackdot] at (0,-4) {};
\node (C) at (0,-5) {};
\node (D) at (1.75,-5) {};
\node (f1) [morphism, wedge] at (0.5,-3) {$f$};
\node (f2) [morphism, wedge, hflip] at (0.5,0) {$f$};
\draw (A.center) to (b1.center);
\draw (b1.center) to [out=right, in=up] (f2.north);
\draw (f2.south) to [out=down, in=left] (w1.center);
\draw (w1.center) to (w2.center);
\draw (w2.center) to [out=left, in=up] (f1.north);
\draw (b2.center) to (C.center);
\draw (w2.center) to [out=right, in=up] (1.5,-3 |- f1.north) to (1.5,-5);
\draw (w1.center) to [out=right, in=down] (1.5,1 |- f2.south) to (1.5,2);
\draw (b1.center) to [out=left, in=up] (-0.5,0 |- f2.north) to (-0.5,0 |- f1.south) to [out=down, in=left] (b2.center);
\draw (f1.south) to [out=down, in=right] (b2.center);
\end{tikzpicture}
\end{aligned}
=\,\,
\ud(A)
\hspace{-3pt}
\begin{aligned}
\begin{tikzpicture}[yscale=0.7,string,xscale=1.2]
\node (f1) [morphism, wedge] at (0.5,-4) {$f$};
\node (f2) [morphism, wedge, hflip] at (-0.5,-2) {$f$};
\node (A) at (-2,0) {};
\node (B) at (0.5,0) {};
\node (w1) [whitedot] at (0.5,-2.0) {};
\node (w2) [whitedot] at (1.0,-3) {};
\node (w3) [whitedot] at (-0.125,-3) {};
\node (b1) [blackdot] at (-1.5,-3) {};
\node (b2) [blackdot] at (0,-5) {};
\node (b3) [blackdot] at (-0.75,-1) {};
\node (C) at (0,-6) {};
\node (D) at (1.5,-6) {};
\draw (A.center) to +(0,-1.5) to [out=down, in=left] (b1.center);
\draw (w1.center) to [out=right, in=up] (w2.center);
\draw (w2.center) to [out=left, in=up] (f1.north);
\draw (f1.south) to [out=down, in=right] (b2.center);
\draw (b2.center) to (C.center);
\draw (w2.center) to [out=right, in=up] (D.center |- f1.north) to (D.center);
\draw (w1.center) to [out=up, in=down] (B.center);
\draw (b1.center) to [out=down, in=left] (b2.center);
\draw (b1.center) to [out=right, in=left] (b3.center);
\draw (f2.south) to [out=down, in=left] (w3.center);
\draw (w3.center) to [out=right, in=left] (w1.center);
\draw (b3.center) to [out=right, in=up] (f2.north);
\end{tikzpicture}
\end{aligned}
=\,\,
\ud(A)
\hspace{-3pt}
\begin{aligned}
\begin{tikzpicture}[yscale=0.7,string,xscale=1.2]
\node (f1) [morphism, wedge] at (0.5,-4) {$f$};
\node (f2) [morphism, wedge] at (-1,-2) {$f$};
\node (A) at (-2,0) {};
\node (B) at (0.68,0) {};
\node (w1) [whitedot] at (0.68,-2.0) {};
\node (w2) [whitedot] at (1.0,-3) {};
\node (w3) [whitedot] at (0.25,-3) {};
\node (b1) [blackdot] at (-1.5,-3) {};
\node (b2) [blackdot] at (0,-5) {};
\node (b3) [graydot] at (-0.5,-1) {};
\node (C) at (0,-6) {};
\node (D) at (1.5,-6) {};
\draw (A.center) to +(0,-1.5) to [out=down, in=left] (b1.center);
\draw (w1.center) to [out=right, in=up] (w2.center);
\draw (w2.center) to [out=left, in=up] (f1.north);
\draw (f1.south) to [out=down, in=right] (b2.center);
\draw (b2.center) to (C.center);
\draw (w2.center) to [out=right, in=up] (D.center |- f1.north) to (D.center);
\draw (w1.center) to [out=up, in=down] (B.center);
\draw (b1.center) to [out=down, in=left] (b2.center);
\draw (w3.center) to [out=left, in=right] (b3.center);
\draw (f2.south) to [out=down, in=right] (b1.center);
\draw (w3.center) to [out=right, in=left] (w1.center);
\draw (b3.center) to [out=left, in=up] (f2.north);
\end{tikzpicture}
\end{aligned}
\\
&
\hspace{2cm}
= \,\,\ud(A)
\begin{aligned}
\begin{tikzpicture}[yscale=0.8,string,xscale=1.2]
\node (f1) [morphism, wedge] at (-0.25,-3) {$f$};
\node (f2) [morphism, wedge] at (1.25,-3) {$f$};
\node (A) at (-1,0) {};
\node (B) at (1.5,0) {};
\node (w1) [whitedot] at (1.5,-1.0) {};
\node (w2) [whitedot] at (1.0,-2) {};
\node (w3) [whitedot] at (0.5,-2.5) {};
\node (b1) [blackdot] at (0.5,-4) {};
\node (b2) [blackdot] at (0,-5) {};
\node (b3) [graydot] at (0,-2) {};
\node (C) at (0,-6) {};
\node (D) at (2.25,-6) {};
\draw (A.center) to [out=down, in=left, in looseness=0.59] (b2.center);
\draw (w1.center) to [out=left, in=up] (w2.center);
\draw (w2.center) to [out=right, in=up] (f2.north);
\draw (b2.center) to (C.center);
\draw (w1.center) to [out=right, in=up, out looseness=0.45] (D.center);
\draw (w1.center) to [out=up, in=down] (B.center);
\draw (b1.center) to [out=down, in=right] (b2.center);
\draw (b1.center) to [out=left, in=down] (f1.south);
\draw (b3.center) to [out=right, in=left] (w3.center);
\draw (w2.center) to [out=left, in=right] (w3.center);
\draw (f1.north) to [out=up, in=left] (b3.center);
\draw (f2.south) to [out=down, in=right] (b1.center);
\end{tikzpicture}
\end{aligned}
= \,\,\ud(A)
\begin{aligned}
\begin{tikzpicture}[yscale=0.8,string,xscale=1.2]
\node (f) [morphism, wedge] at (0.5,-4) {$f$};
\node (A) at (-0.5,0) {};
\node (B) at (1.5,0) {};
\node (w1) [whitedot] at (1.5,-1.0) {};
\node (w2) [whitedot] at (1.0,-2) {};
\node (w3) [whitedot] at (0.5,-2.5) {};
\node (b1) [graydot] at (0.5,-3.25) {};
\node (b2) [blackdot] at (0,-5) {};
\node (b3) [graydot] at (0,-2) {};
\node (C) at (0,-6) {};
\node (D) at (2,-6) {};
\draw (A.center) to +(0,-4) to [out=down, in=left] (b2.center);
\draw (w1.center) to [out=left, in=up] (w2.center);
\draw (w2.center) to [out=right, in=right] (b1.center);
\draw (b2.center) to (C.center);
\draw  (D.center) to +(0,4) to [out=up, in=right] (w1.center);
\draw (w1.center) to [out=up, in=down] (B.center);
\draw (b1.center) to [out=down, in=up] (f.north);
\draw (f.south) to [out=down, in=right] (b2.center);
\draw (b1.center) to [out=left, in=left] (b3.center);
\draw (b3.center) to [out=right, in=left] (w3.center);
\draw (w2.center) to [out=left, in=right] (w3.center);
\end{tikzpicture}
\end{aligned}
\\
&
\hspace{2cm}
=\hspace{5pt}
\begin{aligned}
\begin{tikzpicture}[yscale=0.7,string,xscale=1]
\node (f) [morphism, wedge] at (0.5,-4) {$f$};
\node (A) at (-0.5,0) {};
\node (B) at (1.5,0) {};
\node (w1) [whitedot] at (1.5,-1.0) {};
\node (w2) [whitedot] at (1.0,-2) {};
\node (b1) [graydot] at (0.5,-3) {};
\node (b2) [blackdot] at (0,-5) {};
\node (C) at (0,-6) {};
\node (D) at (2,-6) {};
\draw (A.center) to +(0,-4) to [out=down, in=left] (b2.center);
\draw (w1.center) to [out=left, in=up] (w2.center);
\draw (b2.center) to (C.center);
\draw  (D.center) to +(0,4) to [out=up, in=right] (w1.center);
\draw (w1.center) to [out=up, in=down] (B.center);
\draw (b1.center) to [out=down, in=up] (f.north);
\draw (f.south) to [out=down, in=right] (b2.center);
\end{tikzpicture}
\end{aligned}
\hspace{5pt}=\hspace{5pt}
\begin{aligned}
\begin{tikzpicture}[yscale=0.7,string,xscale=1]
\node (A) at (-0.5,0) {};
\node (B) at (1.5,0) {};
\node (w1) [whitedot] at (1.5,-1.0) {};
\node (w2) [whitedot] at (1.0,-2) {};
\node (b1) [blackdot] at (0.5,-4) {};
\node (b2) [blackdot] at (0,-5) {};
\node (C) at (0,-6) {};
\node (D) at (2,-6) {};
\draw (A.center) to +(0,-4) to [out=down, in=left] (b2.center);
\draw (w1.center) to [out=left, in=up] (w2.center);
\draw (b2.center) to (C.center);
\draw  (D.center) to +(0,4) to [out=up, in=right] (w1.center);
\draw (w1.center) to [out=up, in=down] (B.center);
\draw (b1.center) to [out=down, in=right] (b2.center);
\end{tikzpicture}
\end{aligned}
\hspace{5pt}=\hspace{5pt}
\begin{aligned}
\begin{tikzpicture}[yscale=0.7,string,xscale=1]
\draw (0,0) to (0,6);
\draw (1.5,0) to (1.5,6);
\end{tikzpicture}
\end{aligned}
\end{align*}
There is a similar argument that the other composite also gives the identity. \end{proof}

\section{Identifying group homomorphisms into abelian groups}
\label{sec:algorithm}

\subsection{Introduction}

In this Section we construct a new deterministic quantum algorithm to identify  group homomorphisms.  
\begin{defn}[Group homomorphism identification problem]
Given finite groups $G$ and $A$ where $A$ is abelian, and a blackbox function $f:G\to A$ that is promised to be a group homomorphism, identify the homomorphism $f$.
\end{defn}

\noindent
We will define a quantum algorithm that solves the group homomorphism identification problem with a number of queries equal to the number of simple factors of the abelian group $A$.

For comparison, we can consider the obvious classical algorithm for this problem.
\begin{lemma}
Given finite groups $G$ and $A$, where $A$ is abelian and $G$ has a generating set of order $m$, and a blackbox function $f:G\to A$ that is promised to be a group homomorphism, a classical algorithm can determine $f$ with $m$ oracle queries.
\end{lemma}
\begin{proof}
Once we have evaluated $f$ classically on the generating set of $G$, we have fully characterized~$f$. 
\end{proof}

\noindent
We are unable to prove optimality in either the quantum or classical case. However, we note that the query complexities of these quantum and classical algorithms depend of different and unrelated parameters of the problem. Instances where the order of the generating set of $G$ is larger than the number of factors in the target group $A$ will demonstrate a quantum advantage. 

In  the simpler case where $G$ is an abelian group this quantum algorithm was previously described by H\o yer \cite{hoyer-conjops}, though his algebraic presentation differs significantly from ours. H\o yer also notes that the algorithm by Bernstein and Vazirani in~\cite{bernstein-qcomplex} is an instance of the abelian group identification problem where $G=\mathbb{Z}_n^n$ and $A=\mathbb{Z}_2$. Independently, Cleve et. al.~\cite{cleve-qAlgRevisited} also presented an algorithm for the abelian case where $G=\mathbb{Z}_2^n$ and $A=\mathbb{Z}_2^m$.


We will proceed using the abstract structure defined earlier, but will now work in the dagger-symmetric monoidal category {\bf FHilb}. Any choice of orthonormal basis for  an object $A$ in {\bf FHilb} endows it with a dagger-Frobenius algebra $(A,\tinymult[blackdot],\tinyunit[blackdot])$, whose copying map $d: A \to A\otimes A$ is defined as the linear extension of $d(|i\rangle)=|i\rangle\otimes|i\rangle$. Any finite group $G$ induces a different dagger-Frobenius algebra on an object $A=\mathbb{C}[G]$, the Hilbert space with orthonormal basis given by the elements $G$, with multiplication given by linear extension of the group multiplication; we represent this structure as~$(A, \tinymult[whitedot], \tinyunit[whitedot])$. These two Frobenius algebras are complementary.

\def\Mat{\mathrm{Mat}}
In the case that $G$ is finite, its representations can be characterized as the homomorphisms \mbox{$G \sxto \rho \mbox{Mat}(n)$}. The homomorphism conditions take the following form~\cite[Section~A.7]{vicary-tqa}:
\begin{calign}
\label{eq:rhocopied}
\begin{aligned}
\begin{tikzpicture}[thick, scale=\licsscale]
\draw (-0.7,-1) node [below] {$G$} to [out=up, in=\swangle] (0,0);
\draw (0.7,-1) node [below] {$G$} to [out=up, in=\seangle] (0,0);
\draw (0,0) to (0,0.75);
\node (m) at (0,0) [whitedot] {};
\node (rho) at (0,0.75) [morphism, wedge, width=0, anchor=south] {$\rho$};
\draw ([xshift=5pt] rho.north) to +(0,0.70);
\draw ([xshift=-5pt] rho.north) to +(0,0.70);
\node at (0,2.25) [anchor=south] {$\Mat(n)$};
\end{tikzpicture}
\end{aligned}
\quad=\quad
\begin{aligned}
\begin{tikzpicture}[thick, scale=\licsscale]
\node (r1) at (0,1.5) [morphism, wedge] {$\rho$};
\node (r2) at (1.5,1.5) [morphism, wedge] {$\rho$};
\draw (0,0) node [below] {$G$} to (r1.south);
\draw (1.5,0) node [below] {$G$} to (r2.south);
\draw ([xshift=5pt] r1.north) to [out=up, in=up] ([xshift=-5pt] r2.north);
\draw ([xshift=-5pt] r1.north) to [out=up, in=down, in looseness=1] (0.55,3.25);
\draw ([xshift=5pt] r2.north) to [out=up, in=down, in looseness=1] (0.95,3.25);
\node [above] at (0.75,3.25) {$\Mat(n)$};
\end{tikzpicture}
\end{aligned}
&
\begin{aligned}
\begin{tikzpicture}[thick, scale=\licsscale]
\draw (0,-0.55) node [whitedot] {} to +(0,1) node (r1) [morphism, wedge, anchor=south] {$\rho$};
\draw ([xshift=-5pt] r1.north) to +(0,1);
\draw ([xshift=5pt] r1.north) to +(0,1);
\node at (0.0,2.25) [above] {$\Mat(n)$};
\node [below, white] at (0,-1) {$G$};
\end{tikzpicture}
\end{aligned}
\quad=\quad
\begin{aligned}
\begin{tikzpicture}[thick, scale=\licsscale]
\draw (0,0) to (0,-1) to [out=down, in=down, looseness=2] (0.5,-1) to (0.5,0);
\node at (0.25,0) [above] {$\Mat(n)$};
\node [below, white] at (0.25,-3.25) {$G$};
\end{tikzpicture}
\end{aligned}
\end{calign}
These will be essential for our proofs below.

\subsection{The algorithm}

The structure of the quantum algorithm that solves the group homomorphism identification problem is given by the topological diagram~\eqref{eq:theAlg} below. Here $\sigma:G\to\mathbb{C}$ is a normalized irreducible representation of $G$, representing the result of the measurement, and $\rho:A\to\mathbb{C}$ is a normalized irreducible representation of $A$. The representation $\rho$ is one-dimensional as $A$ is an abelian group. Physically, we are able to produce the input  state $\rho$ efficiently, using $O(\log n)$ time steps, via the quantum Fourier transform for any finite abelian group~\cite{cleve-parallelQFT}. The measurement  result $\sigma$ arises from  a measurement in the Fourier basis, which can, by a similar procedure for any finite group~\cite{childs-qalgebraic}, also be implemented efficiently.
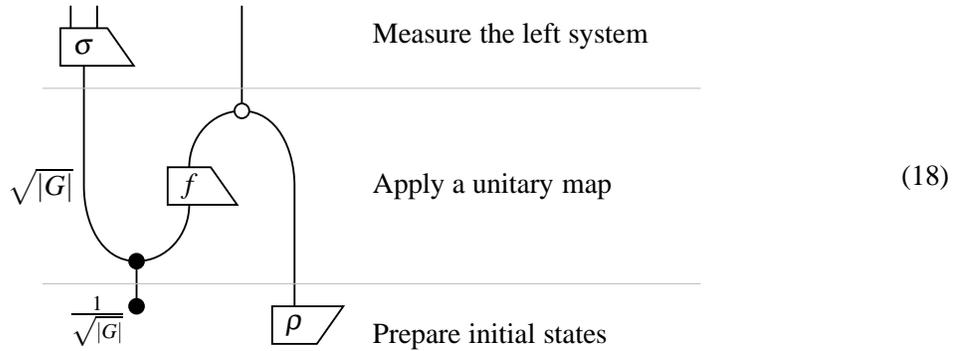
\begin{align}
\label{eq:theAlg}
\begin{aligned}
\begin{tikzpicture}[string, yscale=1]
    \node (dot) [blackdot] at (0,1) {};
    \node (f) [morphism, wedge] at (0.7,2) {$f$};
    \node (m) [whitedot] at (1.4,3) {};
    \node (topsig) [morphism, fill=white, wedge, anchor=south] at (-0.7,3.6) {$\sigma$};
\draw ([xshift=5pt] topsig.north) to +(0,0.3);
\draw ([xshift=-5pt] topsig.north) to +(0,0.3);     
\draw (0,0.4)
        node [blackdot] {}
        node [anchor=20] {$\frac 1 {\sqrt{|G|}}$}
    to (0,1)
    to [out=left, in=south] (-0.7,2)
    to (topsig.south);   
\draw (0,1)
    to [out=right, in=south] (f.south);
\draw  (f.north)
    to [out=up, in=left] (1.4,3)
    to [out=right, in=up] +(0.7,-1)
    to (2.1,0.4)
        node [morphism, wedge, hflip, anchor=north] {$\rho$};
\draw (m.center) to (1.4,4.4)
        node [above] {};
\draw [thin, lightgray] (-1.25,0.7) to (7.5,0.7);
\draw [thin, lightgray] (-1.25,3.3) to (7.5,3.3);
\node at (3,0) [anchor=west] {Prepare initial states};
\node at (3,2) [anchor=west] {Apply a unitary map};
\node at (3,4) [anchor=west] {Measure the left system};
\node at (-0.7,2) [anchor=east] {$\sqrt{|G|}$};
\end{tikzpicture}
\end{aligned}
\end{align}

We can compare the structure of this algorithm to that of the standard quantum algorithm for the hidden subgroup problem. There, the second system is prepared in a state given by the identity element of the group, corresponding to a uniform linear combination of the irreducible representations.  A later measurement of this second system---which is not a part of the standard hidden subgroup algorithm, but can be done without changing the result of the procedure---would collapse this combination to a classical mixture of these representations. The hidden subgroup algorithm therefore contains an amount of classical nondeterminism in its initial setup. In principle removing this, and selecting the input representation strategically, can only improve performance, and we take advantage of this here.

We analyze the effect of our new algorithm as follows.
\begin{lemma}
The algorithm defined by~\eqref{eq:theAlg} gives output $\sigma$ with probability given by the square norm of~$\sigma\circ f^*\circ\rho^*$.
\end{lemma}
\begin{proof}
Using that $\rho$ is a group homomorphism and simple diagrammatic rewrites defined in~\cite[Section~A.9]{vicary-tqa},
 we show the following, making use of the fact that representations are copyable points for group multiplication:
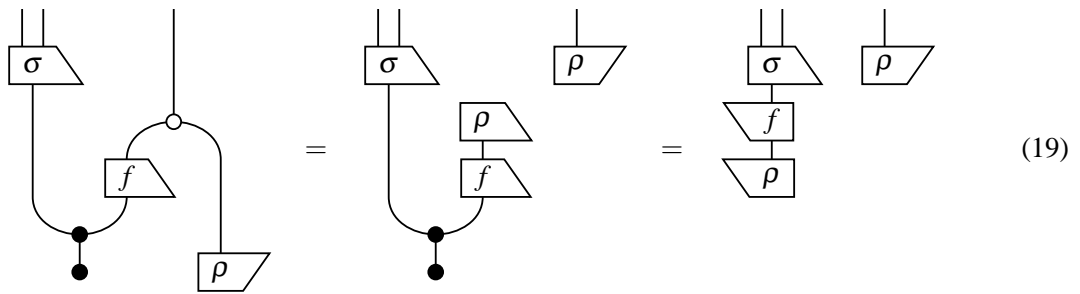
\begin{align}
\label{simplifyAlg}
\begin{aligned}
\begin{tikzpicture}[string]
\draw [use as bounding box, draw=none] (-0.3,0.6) rectangle +(3.45,3.7);
    \node (f) [morphism, wedge] at (1.25,2) {$f$};
    \node (s) [morphism, wedge] at (0,3.5) {$\sigma$};
\node (r) at (2.5,0.75) [morphism, wedge, hflip] {$\rho$};
\draw (f.south) to [out=down, in=right] +(-0.625,-0.5) node (b) [blackdot] {} to [out=left, in=down] +(-0.625, 0.5) to (s.south);
\draw (b.center) to +(0,-0.5) node [blackdot] {};
\draw ([xshift=4pt] s.north) to +(0,0.5);
\draw ([xshift=-4pt] s.north) to +(0,0.5);
\draw (f.north) to [out=up, in=left] +(0.625,0.5) node (w) [whitedot] {} to [out=right, in=up] +(0.625,-0.5) to (r.north);
\draw (w.center) to +(0,1.5);
\end{tikzpicture}
\end{aligned}
\quad=\quad
\begin{aligned}
\begin{tikzpicture}[string]
\draw [use as bounding box, draw=none] (-0.3,0.6) rectangle +(3.45,3.7);
\node (s) [morphism, wedge] at (0,3.5) {$\sigma$};
\node (f) [morphism, wedge] at (1.25,2) {$f$}; 
\node (r) at (1.25, 2.75) [morphism, wedge] {$\rho$};
\node (r2) at (2.5,3.5) [morphism, wedge, hflip] {$\rho$};
\draw (r2.north) to +(0,0.5);
\draw ([xshift=4pt] s.north) to +(0,0.5);
\draw ([xshift=-4pt] s.north) to +(0,0.5);
\draw (r.south) to (f.north);
\draw (f.south) to [out=down, in=right] +(-0.625,-0.5) node (b) [blackdot] {} to [out=left, in=down] +(-0.625, 0.5) to (s.south);
\draw (b.center) to +(0,-0.5) node [blackdot] {};
\end{tikzpicture}
\end{aligned}
\quad=\quad
\begin{aligned}
\begin{tikzpicture}[string]
\draw [use as bounding box, draw=none] (-0.66,0.6) rectangle +(2.75,3.7);
\node (r) [morphism, wedge, hflip, vflip] at (0,2) {$\rho$};
\node (s) [morphism, wedge] at (0,3.5) {$\sigma$};
\node (f) [morphism, wedge, hflip, vflip] at (0,2.75) {$f$};
\node (r2) at (1.5,3.5) [morphism, wedge, hflip] {$\rho$};
\draw ([xshift=4pt] s.north) to +(0,0.5);
\draw ([xshift=-4pt] s.north) to +(0,0.5);
\draw (s.south) to (f.north);
\draw (f.south) to (r.north);
\draw (r2.north) to +(0,0.5);
\end{tikzpicture}
\end{aligned}
\end{align}
The left hand system is thus in the state $\sigma\circ f^*\circ\rho^*$, and using the Born rule, the squared norm of this state gives the probability of this experimental outcome.
\end{proof}

\begin{lemma}\label{lem:irrep}
The composite $\rho\circ f$ is an irreducible representation of $G$.
\end{lemma}
\begin{proof}
The map $f$ is a homomorphism, so $\rho\circ f:G\to\mathbb{C}$ is a one-dimensional representation of $G$. All one-dimensional representations are irreducible, so $\rho\circ f$ is an irreducible representation.
\end{proof}


\begin{lemma}
\label{lem:equaliso}
One-dimensional representations are equivalent only if they are equal.
\end{lemma}
\begin{proof}
Let $\rho_1,\rho_2:G\to \mathbb{C}$ be irreducible representations of $G$. If they are isomorphic, then there exists a linear map $\mathcal{L}:\mathbb{C}\to\mathbb{C}$, i.e. some complex number, such that $\forall g\in G$
$$\mathcal{L}\rho_1(g) = \rho_2(g)\mathcal{L}.$$
Hence we see that $\forall g\in G$, $\rho_1(g) = \rho_2(g)$.
\end{proof}

\begin{theorem}[Structure theorem for finite abelian groups]
\label{thm:structure}
Every finite abelian group is isomorphic to a direct product of cyclic groups of prime power order.
\end{theorem}
\begin{proof}
See~\cite[Theorem 6.4]{artin-algebra} for a proof of this standard result.
\end{proof}

\begin{theorem}\label{rightCyclic}
For a finite group $G$ and cyclic group of prime power order $\mathbb{Z}_{p^n}$, the algorithm~\eqref{eq:theAlg} identifies a group homomorphism $f:G\to \mathbb{Z}_{p^n}$ in a single query.
\end{theorem}
\begin{proof}
Choose the input representation $\rho$ to be the fundamental representation of $\mathbb{Z}_{p^n}$. This representation is faithful.  This means exactly that 
\[ \rho\circ f = \rho\circ f' \qquad \Leftrightarrow \qquad f=f'. \]
Thus $\rho\circ f$ and $\rho\circ f'$ are different irreducible representations if and only if$f$ and $f'$ are different group homomorphisms.  The single measurement on the state $(\rho\circ f)^*$ is performed by the algorithm in the representation basis of $G$, allowing us to determine $\rho\circ f$ up to isomorphism. Due to Lemma~\ref{lem:equaliso} we know that each equivalence class contains only one representative, and thus we can determine $f$ with a single query.
\end{proof}

\begin{theorem}\label{thm:intoAbThm}
For any two finite groups $G$ and $A$, where $A$ is abelian with $n$ simple factors, the quantum algorithm~\eqref{eq:theAlg} can identify a group homomorphism $f:G \to A$ with $n$ oracle queries.
\end{theorem}
\begin{proof}
We prove the result by induction. 
\newline\newline
\noindent{\bf Base case.} When $A=\mathbb{Z}_{p^n}$ is simple, then by Theorem~\ref{rightCyclic} we can identify the homomorphism with a single query.
\newline\newline
\noindent{\bf Inductive step.} If $A$ is not simple, then we must have $A=H_1\times H_2$ by Theorem~\ref{thm:structure}, where the following hold:
\begin{enumerate}

\item The product $\times$ is the direct product whose projectors ($p_1,p_2$) are homomorphisms.

\item  $H_1$ and $H_2$ are groups with $n_1$ and $n_2$ factors respectively such  that the theorem holds, i.e. homomorphisms of the type $f_{1}:G\to H_1$ and $f_{2}:G\to H_2$  can be identified in $n_1$ and $n_2$ queries respectively.

\end{enumerate} 
Since $p_1\circ f$ and $p_2\circ f$ are homomorphisms, we can run subroutines of the algorithm to determine them. Hence we recover $f$ as
\begin{align*}
f(x) = ( (p_1\circ f)(x),(p_2\circ f)(x) ).
\end{align*}
The first subroutine will require $n_1$ queries and the second will require $n_2$ queries, so the total number of queries will be $n_1+n_2$, which is the number of factors of $H_1\times H_2$.
\end{proof}

\ignore{In practice we run the algorithm once on each $k$-th factor of $A$ to determine homomorphisms 
\[ p_k\circ f = f_k:G\to H_k \]
We then know $f$ as 
\begin{align*}
f(n_0,n_1,...n_{n-1}) = \left(f_0(n_0), f_1(n_1),..., f_{n-1}(n_{n-1})\right)
\end{align*}
we know that each isomorphic class contains only one representation}

\subsection{Extension to the non-abelian case}
We now consider the more general case where the target group $A$ is non-abelian. We do not know how to extend the algorithm described above to this case.  Nevertheless, it is instructive to analyze this scenario in our graphical approach. 

Irreducible representations of a non-abelian group $A$ are not necessarily one dimensional, though we are still able to compute them via the Fourier transform efficiently \cite{childs-qalgebraic}. In this case the algorithm has the following structure, where $\psi$ represents the initial state of the right-hand system in the representation space:
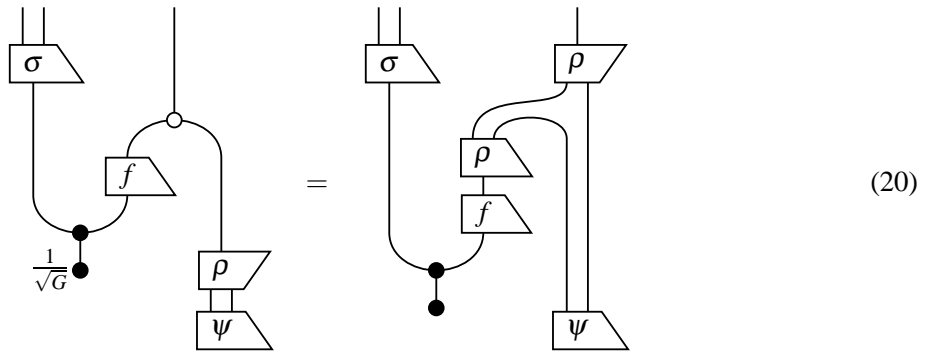
\begin{equation}
\label{eq:NonAbAlg}
\begin{aligned}
\begin{tikzpicture}[string]
\draw [use as bounding box, draw=none] (-0.3,-0.3) rectangle +(3.45,4.55);
\node (f) [morphism, wedge] at (1.25,2) {$f$};
\node (topsig) [morphism, wedge] at (0,3.5) {$\sigma$};
\node (r) at (2.5,0.75) [morphism, wedge, hflip] {$\rho$};
\draw (f.south) to [out=down, in=right] +(-0.625,-0.5) node (b) [blackdot] {} to [out=left, in=down] +(-0.625, 0.5) to (s.south);
\draw (b.center) to +(0,-0.5) node [blackdot] {} node [anchor=east] {$\frac {1} {\sqrt{G}}$};
\draw ([xshift=4pt] s.north) to +(0,0.5);
\draw ([xshift=-4pt] s.north) to +(0,0.5);
\draw (f.north) to [out=up, in=left] +(0.625,0.5) node (w) [whitedot] {} to [out=right, in=up] +(0.625,-0.5) to (r.north);
\draw (w.center) to +(0,1.5);
\draw ([xshift=4pt] r.south) to +(0,-0.3);
\draw ([xshift=-4pt] r.south) to +(0,-0.3);
\node [morphism, wedge, anchor=north] at ([yshift=-0.3cm] r.south) {$\psi$};
\end{tikzpicture}
\end{aligned}
\quad=\quad
\begin{aligned}
\begin{tikzpicture}[string]
\draw [use as bounding box, draw=none] (-0.3,-0.3) rectangle +(3.45,4.55);
\node (s) [morphism, wedge] at (0,3.5) {$\sigma$};
\node (f) [morphism, wedge] at (1.25,1.5) {$f$}; 
\node (r) at (1.25, 2.25) [morphism, wedge] {$\rho$};
\node (r2) at (2.5,3.5) [morphism, wedge, hflip] {$\rho$};
\draw (r2.north) to +(0,0.5);
\draw ([xshift=4pt] s.north) to +(0,0.5);
\draw ([xshift=-4pt] s.north) to +(0,0.5);
\draw (r.south) to (f.north);
\draw (f.south) to [out=down, in=right] +(-0.625,-0.5) node (b) [blackdot] {} to [out=left, in=down] +(-0.625, 0.5) to (s.south);
\draw (b.center) to +(0,-0.5) node [blackdot] {};
\node (psi) [morphism, wedge, anchor=north] at (2.5,0.2) {$\psi$};
\draw ([xshift=4pt] psi.north) to ([xshift=4pt] r2.south);
\draw ([xshift=-4pt] psi.north) to ([xshift=-4pt] r.north -| r2.south) to [out=up, in=up] ([xshift=4pt] r.north);
\draw ([xshift=-4pt] r.north) to [out looseness=1.3, out=up, in=down, in looseness=0.8] ([xshift=-4pt] r2.south);
\end{tikzpicture}
\end{aligned}
\end{equation}
We notice two additional features in this case. First, it is clear that the left and right systems are no longer in a product state at the end of the protocol, as they were in the final diagram of \eqref{simplifyAlg}. Second, we now have an additional choice when preparing the input representation $\rho$; in order to construct a state from a representation $\rho$ we also must choose the state $\psi$.

While this provides a clear description of the algorithm in this more general setting, it is not clear that it would identify homomorphisms into non-abelian groups. Complications include the lack of a structure theorem that satisfies the conditions for Theorem~\ref{thm:intoAbThm}, and that Lemma~\ref{lem:irrep} no longer applies.  In this setting it may be useful to make the problem easier by restricting to the identification of homomorphisms up to \emph{natural isomorphism}, i.e. where two homomorphisms $f_1,f_2:G\to H$ are considered equivalent when there exists some $\eta\in H$ such that, for all $g\in G$, we have $\eta f_1(g) \eta^{-1} = f_2(g)$.

\section{Application to signal-flow calculus}
\label{sec:signalflow}

\subsection{Introduction}
\label{sec:signalintroduction}

Signal-flow diagrams are a notation in electrical engineering that describe the flow of information in electrical circuits, including rich phenomena such as feedback. Various authors~\cite{fong-transfer, baezerbele, sobocinski}  have have developed a categorical approach to modelling signal-flow diagrams, based on a category of linear relations on vector spaces over a field $k$. We show in this Section that unitary oracles exist in their setup, in the sense of our Definition~\ref{oracle}, and discuss the consequences of this. 

We begin with a brief introduction to the theory, following the terminology of~\cite{baezerbele}.
\def\sto{\rightsquigarrow}
\begin{defn}
The category $\cat{FinRel}_k$ of \textit{linear relations} is defined in the following way, for any field~$k$:
\begin{itemize}
\item \textbf{Objects} are finite dimensional $k$-vector spaces
\item A \textbf{morphism} $f:V \sto W$ is a \textit{linear relation}, defined as a subspace  $S_f \hookrightarrow V \oplus W$
\item \textbf{Composition} of linear relations $f:U \sto V$ and $g: V \sto W$ is defined as the following subspace of $U \oplus W$:
\begin{equation}
\{ (u,w) | \exists v \in V \text{ with } (u,v) \in S_f \text{ and } (v,w) \in S_g \}
\end{equation}
It can be verified that this defines a linear subspace of $U \oplus W$.
\end{itemize}
Note that a linear relation is in particular an ordinary relation, and that composition of linear relations is the same as for ordinary relations. The category $\cat{FinRel}_k$ can be given a monoidal structure in a natural way, using the direct sum of vector spaces.
\end{defn}

For every linear relation, we can define a converse as follows.
\begin{defn}
\label{def:converse}
Given a linear relation $f: U \sto V$ defined as the subspace $S_f \hookrightarrow U \oplus V$, its \emph{converse} is the linear relation $f ^\dag : V \sto U$ defined as the subspace $S_f \hookrightarrow U \oplus V \sxto{\text{swap}} V \oplus U$.
\end{defn}
\noindent
This makes $\cat{FinRel}_k$ into a monoidal dagger-category. Following the usual convention~\cite{selinger}, we depict the dagger of a linear relation as the original morphism flipped about a horizontal axis.

Certain canonical linear relations play an important role in the theory. We define them here, along with the graphical symbol we will use to denote them.
\def\br{\text{\textit{\textbf{r}}}}
\begin{defn}
\label{defn:basicrelations}
The \textit{addition}, \textit{zero}, \textit{copying}, \textit{deletion} and \textit{multiplier} linear relations are defined as follows, where the definitions in the last line are valid for all $a,b \in k$, and where the multiplier relation takes a parameter given by some $r \in k$:
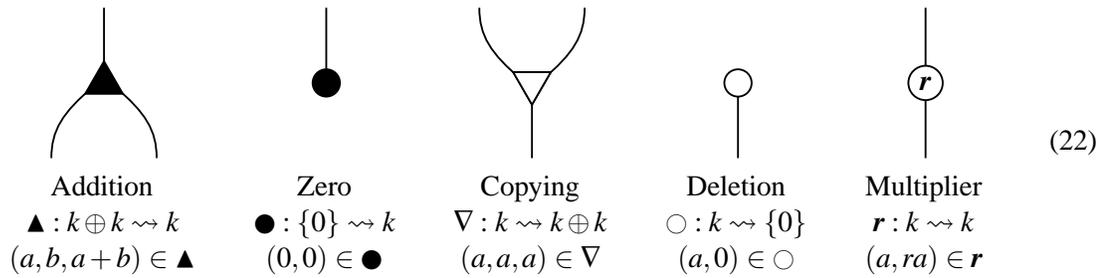
\begin{equation}
\begin{array}{c@{\qquad}c@{\qquad}c@{\qquad}c@{\qquad}c}
\begin{aligned}
\begin{tikzpicture}[string]
\node (n) [uptriangle] at (0,0) {};
\draw (0,1) to (0,0);
\draw [shorten >=-5pt] (-0.7,-1) to [out=up, in=-140] (n.corner 2);
\draw [shorten >=-5pt] (0.7,-1) to [out=up, in=-40] (n.corner 3);
\end{tikzpicture}
\end{aligned}
&
\begin{aligned}
\begin{tikzpicture}[string]
\draw [white] (0,-1) to (0,1);
\node [circle, draw, fill=black, minimum width=10pt] at (0,0) {};
\draw (0,0) to (0,1);
\end{tikzpicture}
\end{aligned}
&
\begin{aligned}
\begin{tikzpicture}[string]
\node (n) [downtriangle, fill=white] at (0,0) {};
\draw (0,-1) to (0,0);
\draw [shorten >=-5pt] (-0.7,1) to [out=down, in=140] (n.corner 3);
\draw [shorten >=-5pt] (0.7,1) to [out=down, in=40] (n.corner 2);
\node (n) [downtriangle, fill=white] at (0,0) {};
\end{tikzpicture}
\end{aligned}
&
\begin{aligned}
\begin{tikzpicture}[string]
\draw [white] (0,-1) to (0,1);
\draw (0,0) to (0,-1);
\node [circle, draw, fill=white, minimum width=10pt] at (0,0) {};
\end{tikzpicture}
\end{aligned}
&
\begin{aligned}
\begin{tikzpicture}[string]
\draw (0,-1) to (0,1);
\node [circle, draw, inner sep=0pt, minimum width=13pt, fill=white] at (0,0) {\br};
\end{tikzpicture}
\end{aligned}
\\
\text{Addition}
&
\text{Zero}
&
\text{Copying}
&
\text{Deletion}
&
\text{Multiplier}
\\
\blacktriangle : k \oplus k \sto k
& \newmoon : \{0\} \sto k
& \nabla : k \sto k \oplus k
& \ocircle: k \sto \{0\}
& \br:k \sto k
\\
(a,b,a+b) \in \blacktriangle
& (0,0) \in \newmoon
& (a,a,a) \in \nabla
& (a,0) \in \ocircle
& (a,ra) \in \br
\end{array}
\end{equation}
\end{defn}

\noindent
They use their theory to model resistors in electrical circuits, using the following network:
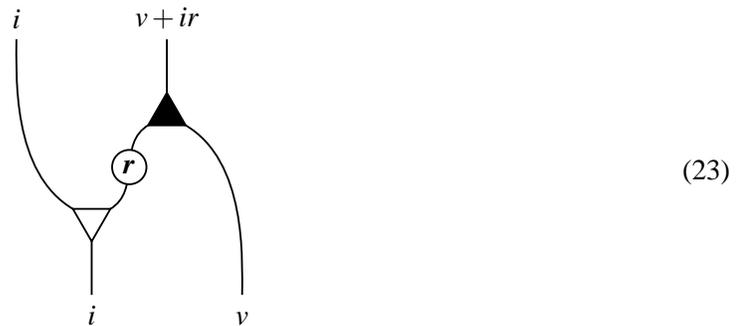
\begin{equation}
\label{eq:resistor}
\begin{aligned}
\begin{tikzpicture}[string, yscale=1, xscale=1]
\node (black) [uptriangle] at (0,1.2) {};
\node (white) [downtriangle, fill=white] at (-1,-0.2) {};
\draw [shorten <=-1pt, shorten >=-1pt] (black.corner 2) to [out=-150, in=30] (white.corner 2);
\draw [shorten >=-1pt] (1,-1.2) node [below] {$v\vphantom{i}$} to [out=up, in=-30, out looseness=0.5] (black.corner 3);
\draw [shorten <=-1pt] (white.corner 1) to (-1,-1.2) node [below] {$i$};
\draw [shorten >=-1pt] (-2,2.2) node [above] {$i$} to [out=down, in=150, out looseness=0.5] (white.corner 3);
\draw [shorten <=-1pt] (black.corner 1) to (0,2.2) node [above] {$v+ir$};
\node (r) [circle, draw, inner sep=0pt, minimum width=13pt, fill=white] at (-0.5,0.5) {\br};
\end{tikzpicture}
\end{aligned}
\end{equation}
The left-hand wire represents the current variable, and the right-hand wire represents the voltage variable. The initial current-voltage pair $(i,v)$ is mapped to the output current-voltage pair $(i,v+ir)$. This respects the usual law for resistors in electrical circuits, whereby if $\delta v$ is the change in voltage over a resistor, $i$ is the current through the resistor, and the value of the resistance is $r$, then $\delta v = i r$.

It has been recognized in~\cite{baezerbele} that the linear relations given in Definition~\ref{defn:basicrelations} satisfy many interesting relationships, which we summarize here without proof:
\begin{lemma}
\label{lem:initialproperties}
In $\cat{FinRel}_k$, the following relationships hold between the addition, zero, copying, deletion and multiplier linear relations:
\begin{enumerate}
\item Addition and zero together form a commutative monoid.
\item Copying and deletion together form a commutative comonoid.
\item This monoid and comonoid together form a bialgebra.
\item The multiplier relation is a monoid homomorphism for addition, and a comonoid homomorphism for copying.
\end{enumerate} 
\end{lemma}

\subsection{Complementary dagger-Frobenius structure}

In this section we prove new results about the structures introduced in Section~\ref{sec:signalintroduction}. We begin by establishing the existence of dagger-Frobenius properties of the addition and copying operations.
\begin{lemma}
\label{lem:signalfrobenius}
In $\cat{FinRel}_k$, the addition and copying linear relations separately form commutative dagger-Frobenius algebras.
\end{lemma}
\begin{proof}
That addition and zero forms a commutative monoid, and copying and deletion forms a commutative comonoid, is established in Lemma~\ref{lem:initialproperties}. It remains to demonstrate that the dagger-Frobenius conditions hold for each of these structures.

We first evaluate the action of the following composite linear relation, which is one side of the  dagger-Frobenius condition for the copying linear relation:
\begin{equation}
\begin{aligned}
\begin{tikzpicture}[string, scale=1, yscale=1]
\node (black) [uptriangle, fill=white] at (0,1) {};
\node (white) [downtriangle, fill=white] at (1,0.2) {};
\draw [shorten <=-1pt, shorten >=-1pt] (black.corner 3) to [out=-30, in=150] (white.corner 3);
\draw [shorten >=-1pt] (-0.8,-0.5) to [out=up, in=-150] (black.corner 2);
\draw [shorten <=-1pt] (white.corner 1) to (1,-0.5);
\draw [shorten >=-1pt] (1.8,1.7) to [out=down, in=30] (white.corner 2);
\draw [shorten <=-1pt] (black.corner 1) to (0,1.7);
\node [below] at (-0.8,-0.5) {$a\vphantom{|}$};
\node [below] at (1,-0.5) {$b$};
\node [below] at (0.3,0.8) {$b$};
\node [above] at (0,1.7) {$a$};
\node [above] at (-1,1.63) {(if $a=b$)};
\node [above] at (1.8,1.7) {$b$};
\end{tikzpicture}
\end{aligned}
\end{equation}
We see that this composite relation can be defined as  $\forall a, (a,a) \smallwhitediagram (a,a)$, and similarly it can be shown that $\forall a, (a,a) \smallwhitediagramflip (a,a)$. Hence we have demonstrated the dagger-Frobenius condition $\smallwhitediagram = \smallwhitediagramflip$.

For the addition linear relation, we calculate the left side of the dagger-Frobenius condition as follows:
\begin{equation}
\begin{aligned}
\begin{tikzpicture}[string, scale=1.1, yscale=1.1]
\node (black) [uptriangle] at (0,1) {};
\node (white) [downtriangle] at (1,0.2) {};
\draw [shorten <=-1pt, shorten >=-1pt] (black.corner 3) to [out=-30, in=150] (white.corner 3);
\draw [shorten >=-1pt] (-0.8,-0.5) to [out=up, in=-150] (black.corner 2);
\draw [shorten <=-1pt] (white.corner 1) to (1,-0.5);
\draw [shorten >=-1pt] (1.8,1.7) to [out=down, in=30] (white.corner 2);
\draw [shorten <=-1pt] (black.corner 1) to (0,1.7);
\node [below] at (-0.8,-0.5) {$a\vphantom{|}$};
\node [below] at (1,-0.5) {$b$};
\node [left] at (0.65,0.4) {$\forall c, c$};
\node [above] at (0,1.7) {$a+c$};
\node [above] at (-1,1.63) {$\forall c,$};
\node [above] at (1.8,1.7) {$b-c$};
\end{tikzpicture}
\end{aligned}
\end{equation}
We can write this action succinctly as $\forall c,(a,b) \smallblackdiagram (a+c,b-c)$. Similarly, the other composite can be shown to have action $\forall c,(a,b) \smallblackdiagramflip (a-c,b+c)$. Making the substitution $c':=-c$, we can rewrite this second definition as $\forall c',(a,b) \smallblackdiagramflip (a+c',b-c')$. This demonstrates that $\smallblackdiagram = \smallblackdiagramflip$ as linear relations, verifying the dagger-Frobenius condition for the addition linear relation.
\end{proof}

Furthermore, these Frobenius algebras interact as complementary structures.
\begin{lemma}
In $\cat{FinRel}_k$, the addition and copying linear relations form complementary dagger-Frobenius algebras.
\end{lemma}
\begin{proof}
We have already established the Frobenius properties in Lemma~\ref{lem:signalfrobenius}. It remains to demonstrate the complementarity condition.

We evaluate the action of the following composite relation:
\begin{equation}
\begin{aligned}
\begin{tikzpicture}[string]
\node (black) [uptriangle] at (0,1) {};
\node (white) [downtriangle, fill=white] at (1,0.2) {};
\draw [shorten <=-1pt, shorten >=-1pt] (black.corner 3) to [out=-30, in=150] (white.corner 3);
\draw [shorten >=-1pt] (-0.8,-0.5) to [out=up, in=-150] (black.corner 2);
\draw [shorten <=-1pt] (white.corner 1) to (1,-0.5);
\draw [shorten >=-1pt] (1.8,1.7) to [out=down, in=30] (white.corner 2);
\draw (0,1) to (0,1.7);
\node [below] at (-0.8,-0.5) {$a$};
\node [below] at (1,-0.5) {$b$};
\node [below] at (0.3,0.8) {$b$};
\node [above] at (0,1.7) {$a+b$};
\node [above] at (1.8,1.7) {$b$};
\end{tikzpicture}
\end{aligned}
\end{equation}
Writing $K$ for this linear relation, we see that $K$ is given by $\forall a,b\in k, (a,b) K (a+b,b)$. By Definition~\ref{def:converse} of the converse relation, we see that $K ^\dag$ is defined as $\forall a,b\in k,(a+b,b)K^\dag (a,b)$, or equivalently $\forall a,b \in k,(a,b) K^\dag (a-b,b)$. Since $K$ is single-valued and total, it is clear that $K$ and $K^\dag$ are inverse, as can be shown by explicit calculation. By Theorem~\ref{thm:complementarityunitary}, it follows that addition and copying are complementary. 
\end{proof}

The final property that we establish is that multipliers are self-conjugate.
\begin{lemma}
In $\cat{FinRel}_k$, a multiplier $\br : k \sto k$ is a self-conjugate morphism.
\end{lemma}
\begin{proof}
We must verify that $\br$ is equal to the transpose of its dagger:
\begin{equation}
\begin{aligned}
\begin{tikzpicture}[string, yscale=1.25]
\draw (0,0) node [below] {$a$} to node [circle, draw, inner sep=0pt, minimum width=13pt, fill=white] {\br} (0,3.4) node [above] {$ra$};
\end{tikzpicture}
\end{aligned}
\qquad=\qquad
\begin{aligned}
\begin{tikzpicture}[string, yscale=1.25]
\draw [lightgray] (-1.3,-0.8) to (2.3,-0.8);
\draw [lightgray] (-1.3,-0.2) to (2.3,-0.2);
\draw [lightgray] (-1.3,0.5) to (2.3,0.5);
\draw [lightgray] (-1.3,1.2) to (2.3,1.2);
\draw [lightgray] (-1.3,1.8) to (2.3,1.8);
\node (black) [uptriangle] at (0,1.2) {};
\node (white) [downtriangle] at (1,-0.2) {};
\draw [shorten <=-1pt, shorten >=-1pt] (black.corner 3) to [out=-30, in=150] (white.corner 3);
\draw [shorten >=-1pt] (-1,-1.2) to [out=up, in=-150, out looseness=0.5] (black.corner 2);
\draw [shorten <=-1pt] (white.corner 1) to (1,-0.8);
\draw [shorten >=-1pt] (2,2.2) to [out=down, in=30, out looseness=0.5] (white.corner 2);
\draw [shorten <=-1pt] (black.corner 1) to (0,1.7);
\node (r) [circle, draw, inner sep=0pt, minimum width=13pt, fill=white] at (0.5,0.5) {$\br \smash{\scriptstyle {}^\dagger}$};
\node [below] at (-1,-1.2) {$a$};
\node [left] at (0.55,0.1) {$b\vphantom{,|}$};
\node at (1.3,-0.55) {$0\vphantom{,|}$};
\node at (1.9,0.1) {$-b\vphantom{,|}$};
\node [left] at (-1.2,0.1) {$\forall b,\vphantom{,|}$};
\node [left] at (-0.8,0.1) {$a\vphantom{,|}$};
\node [left] at (-1.2,0.8) {$\forall b,\vphantom{,|}$};
\node [left] at (1.3,0.8) {$b/r\vphantom{,|}$};
\node [left] at (1.5,1.5) {$a+b/r\vphantom{,|}$};
\node at (2.2,0.8) {$-b\vphantom{,|}$};
\node at (2.35,1.5) {$-b\vphantom{,|}$};
\node [left] at (-0.5,0.8) {$a\vphantom{,|}$};
\node [left] at (-1.2,1.5) {$\forall b,\vphantom{,|}$};
\node [above] at (2,2.2) {$ra$};
\node [circle, draw, fill=black] at (0,1.8) {};
\node [circle, draw, fill=black] at (1,-0.8) {};
\end{tikzpicture}
\end{aligned}
\end{equation}
On the right-hand side we see that $a$ is related to $-b$, with the constraint that $a+b/r=0$, i.e. that $-b=ra$. This is equal as a linear relation to that of $\br$ itself, given on the left-hand side. This establishes the result.
\end{proof}

Given these results, we are motivated to make the following definitions which generalize the motivating example of the theory of signal-flow diagrams in $\cat{FinRel}_k$.
\begin{defn}
In a symmetric monoidal dagger-category, a \emph{signal-flow structure} is an object $A$ equipped with a pair of commutative dagger-Frobenius algebras, which interact as a bialgebra. A \emph{multiplier} for this signal-flow structure is a self-conjugate morphism $\br:A \to A$ which is a monoid and comonoid homomorphism for both structures.
\end{defn}
\begin{defn}
Given a signal-flow structure equipped with a multiplier $\br$, the \emph{resistor} associated to $\br$ is the composite given by diagram~\eqref{eq:resistor}.
\end{defn}

\noindent
We then apply our earlier result to show that resistors are always unitary.
\begin{corollary}
Given a signal-flow structure equipped with a multiplier, its resistor is unitary.
\end{corollary}
\begin{proof}
An immediate application of Theorem~\ref{thm:familyofunitaries}.
\end{proof}

\noindent
The appearance of this unitary structure in both quantum algorithm and the signal-flow calculus highlights the general role that this abstract structure can play in different process theories.

\bibliographystyle{eptcs}

\end{document}

%% file: header.tex
\usepackage{docmute}
\usepackage{enumerate,multirow}
\usepackage{amsmath}
\usepackage{amsthm,amsfonts,amssymb,mathtools,xspace}
\usepackage{tikz}
\usetikzlibrary{decorations.pathreplacing}
\usetikzlibrary{decorations.markings}
\usetikzlibrary{calc}
\usetikzlibrary{shapes.geometric}
\usepackage{etoolbox}
\usepackage{datetime}

\theoremstyle{plain}
\newtheorem{theorem}{Theorem}
\newtheorem{lemma}[theorem]{Lemma}
\newtheorem{corollary}[theorem]{Corollary}

\theoremstyle{definition}
\newtheorem{defn}[theorem]{Definition}

\newenvironment{pic}[1][]
{\begin{aligned}\begin{tikzpicture}[#1]}
{\end{tikzpicture}\end{aligned}}
\newcommand{\edges}[1][]%
{
}

\makeatletter
\def\calign@preamble{%
   &\hfil\strut@
    \setboxz@h{\@lign$\m@th\displaystyle{##}$}%
    \ifmeasuring@\savefieldlength@\fi
    \set@field
    \hfil
    \tabskip\alignsep@
}
\let\cmeasure@\measure@
\patchcmd\cmeasure@{\divide\@tempcntb\tw@}{}{}{}
\patchcmd\cmeasure@{\divide\@tempcntb\tw@}{}{}{}
\patchcmd\cmeasure@{\ifodd\maxfields@
  \global\advance\maxfields@\@ne
  \fi}{}{}{}
\newenvironment{calign}
{%
  \let\align@preamble\calign@preamble
  \let\measure@\cmeasure@
  \align
}
{%
  \endalign
}
\makeatother


\newcommand\tinymatrix[1]
{\left( \hspace{-2pt} \renewcommand\thickspace{\kern2pt} \scriptstyle\begin{smallmatrix} #1 \end{smallmatrix} \hspace{-2pt} \right)}

\newcommand\ignore[1]{}

\pgfdeclarelayer{foreground}
\pgfdeclarelayer{background}
\pgfsetlayers{main,foreground,background}
\tikzset{smallbox/.style={draw, fill=white, minimum height=0.45cm, minimum width=0.45cm, inner sep=-100pt}}

\makeatletter
\pgfkeys{%
  /tikz/on layer/.code={
    \pgfonlayer{#1}\begingroup
    \aftergroup\endpgfonlayer
    \aftergroup\endgroup
  },
  /tikz/node on layer/.code={
    \gdef\node@@on@layer{%
      \setbox\tikz@tempbox=\hbox\bgroup\pgfonlayer{#1}\unhbox\tikz@tempbox\endpgfonlayer\egroup}
    \aftergroup\node@on@layer
  },
  /tikz/end node on layer/.code={
    \endpgfonlayer\endgroup\endgroup
  }
}
\def\node@on@layer{\aftergroup\node@@on@layer}
\makeatother\def\thickness{0.7pt}
\tikzstyle{oldmorphism}=[minimum width=30pt, minimum height=16pt, draw, font=\small, inner sep=0pt, fill=white, line width=\thickness]
\tikzstyle{cross}=[preaction={draw=white, -, line width=10pt}]
\tikzstyle{braid}=[double=black, line width=3*\thickness, double distance=\thickness, white]
\tikzstyle{string}=[line width=\thickness]
\tikzstyle{scalar}=[circle, inner sep=0pt, minimum width=15pt, draw, line width=\thickness]
\tikzstyle{dot}=[circle, draw=black, fill=black!25, inner sep=.4ex, line width=\thickness, node on layer=foreground]
\tikzstyle{blackdot}=[circle, draw=black, fill=black!100, inner sep=.4ex, line width=\thickness, node on layer=foreground]
\tikzstyle{graydot}=[circle, draw=black, fill=gray!40!white, inner sep=.4ex, line width=\thickness, node on layer=foreground]
\tikzstyle{whitedot}=[circle, draw=black, fill=white, inner sep=.4ex, line width=\thickness, node on layer=foreground]
\tikzstyle{mixedmorphism}=[morphism, minimum width=30pt, minimum height=16pt, draw, font=\small, inner sep=0pt, fill=white, line width=\thickness,rounded corners=1ex]
\tikzstyle{thick}=[line width=\thickness]
\tikzstyle{tiny}=[font=\tiny]

\tikzset{arrow/.style={decoration={
    markings,
    mark=at position #1 with \arrow{thickarrow}},
    postaction=decorate}
}
\tikzset{reverse arrow/.style={decoration={
    markings,
    mark=at position #1 with \arrow{reversethickarrow}},
    postaction=decorate}
}

\newif\ifblack\pgfkeys{/tikz/black/.is if=black}
\newif\ifwedge\pgfkeys{/tikz/wedge/.is if=wedge}
\newif\ifvflip\pgfkeys{/tikz/vflip/.is if=vflip}
\newif\ifhflip\pgfkeys{/tikz/hflip/.is if=hflip}
\newif\ifhvflip\pgfkeys{/tikz/hvflip/.is if=hvflip}
\newif\ifconnectnw\pgfkeys{/tikz/connect nw/.is if=connectnw}
\newif\ifconnectne\pgfkeys{/tikz/connect ne/.is if=connectne}
\newif\ifconnectsw\pgfkeys{/tikz/connect sw/.is if=connectsw}
\newif\ifconnectse\pgfkeys{/tikz/connect se/.is if=connectse}
\newif\ifconnectn\pgfkeys{/tikz/connect n/.is if=connectn}
\newif\ifconnects\pgfkeys{/tikz/connect s/.is if=connects}
\newif\ifconnectnwf\pgfkeys{/tikz/connect nw >/.is if=connectnwf}
\newif\ifconnectnef\pgfkeys{/tikz/connect ne >/.is if=connectnef}
\newif\ifconnectswf\pgfkeys{/tikz/connect sw >/.is if=connectswf}
\newif\ifconnectsef\pgfkeys{/tikz/connect se >/.is if=connectsef}
\newif\ifconnectnf\pgfkeys{/tikz/connect n >/.is if=connectnf}
\newif\ifconnectsf\pgfkeys{/tikz/connect s >/.is if=connectsf}
\newif\ifconnectnwr\pgfkeys{/tikz/connect nw </.is if=connectnwr}
\newif\ifconnectner\pgfkeys{/tikz/connect ne </.is if=connectner}
\newif\ifconnectswr\pgfkeys{/tikz/connect sw </.is if=connectswr}
\newif\ifconnectser\pgfkeys{/tikz/connect se </.is if=connectser}
\newif\ifconnectnr\pgfkeys{/tikz/connect n </.is if=connectnr}
\newif\ifconnectsr\pgfkeys{/tikz/connect s </.is if=connectsr}
\tikzset{keylengthnw/.initial=\connectheight}
\tikzset{keylengthn/.initial =\connectheight}
\tikzset{keylengthne/.initial=\connectheight}
\tikzset{keylengthsw/.initial=\connectheight}
\tikzset{keylengths/.initial =\connectheight}
\tikzset{keylengthse/.initial=\connectheight}
\tikzset{connect nw length/.style={connect nw=true, keylengthnw={#1}}}
\tikzset{connect n length/.style ={connect n =true, keylengthn ={#1}}}
\tikzset{connect ne length/.style={connect ne=true, keylengthne={#1}}}
\tikzset{connect sw length/.style={connect sw=true, keylengthsw={#1}}}
\tikzset{connect s length/.style ={connect s =true, keylengths ={#1}}}
\tikzset{connect se length/.style={connect se=true, keylengthse={#1}}}
\tikzset{connect nw < length/.style={connect nw <=true, keylengthnw={#1}}}
\tikzset{connect n < length/.style ={connect n <=true,  keylengthn ={#1}}}
\tikzset{connect ne < length/.style={connect ne <=true, keylengthne={#1}}}
\tikzset{connect sw < length/.style={connect sw <=true, keylengthnw={#1}}}
\tikzset{connect s < length/.style ={connect s <=true,  keylengths ={#1}}}
\tikzset{connect se < length/.style={connect se <=true, keylengthse={#1}}}
\tikzset{connect nw > length/.style={connect nw >=true, keylengthnw={#1}}}
\tikzset{connect n > length/.style ={connect n >=true,  keylengthn ={#1}}}
\tikzset{connect ne > length/.style={connect ne >=true, keylengthne={#1}}}
\tikzset{connect sw > length/.style={connect sw >=true, keylengthsw={#1}}}
\tikzset{connect s > length/.style ={connect s >=true,  keylengths ={#1}}}
\tikzset{connect se > length/.style={connect se >=true, keylengthse={#1}}}

\newlength\morphismheight
\newlength\minimummorphismwidth
\newlength\stateheight
\newlength\minimumstatewidth
\newlength\connectheight
\tikzset{width/.initial=\minimummorphismwidth}

\makeatletter
\pgfarrowsdeclare{thickarrow}{thickarrow}
{
  \pgfutil@tempdima=-0.84pt%
  \advance\pgfutil@tempdima by-1.3\pgflinewidth%
  \pgfutil@tempdimb=-1.7pt%
  \advance\pgfutil@tempdimb by.625\pgflinewidth%
  \pgfarrowsleftextend{+\pgfutil@tempdima}
  \pgfarrowsrightextend{+\pgfutil@tempdimb}
}
{
  \pgfmathparse{\pgfgetarrowoptions{thickarrow}}%
  \pgfsetlinewidth{1.25 pt}
  \pgfutil@tempdima=0.28pt%
  \advance\pgfutil@tempdima by.3\pgflinewidth%
  \pgfsetlinewidth{0.8\pgflinewidth}
  \pgfsetdash{}{+0pt}
  \pgfsetroundcap
  \pgfsetroundjoin
  \pgfpathmoveto{\pgfqpoint{-3\pgfutil@tempdima}{4\pgfutil@tempdima}}
  \pgfpathcurveto
  {\pgfqpoint{-2.75\pgfutil@tempdima}{2.5\pgfutil@tempdima}}
  {\pgfqpoint{0pt}{0.25\pgfutil@tempdima}}
  {\pgfqpoint{0.75\pgfutil@tempdima}{0pt}}
  \pgfpathcurveto
  {\pgfqpoint{0pt}{-0.25\pgfutil@tempdima}}
  {\pgfqpoint{-2.75\pgfutil@tempdima}{-2.5\pgfutil@tempdima}}
  {\pgfqpoint{-3\pgfutil@tempdima}{-4\pgfutil@tempdima}}
  \pgfusepathqstroke
}
\pgfarrowsdeclare{reversethickarrow}{reversethickarrow}
{
  \pgfutil@tempdima=-0.84pt%
  \advance\pgfutil@tempdima by-1.3\pgflinewidth%
  \pgfutil@tempdimb=0.2pt%
  \advance\pgfutil@tempdimb by.625\pgflinewidth%
  \pgfarrowsleftextend{+\pgfutil@tempdima}
  \pgfarrowsrightextend{+\pgfutil@tempdimb}
}
{
  \pgftransformxscale{-1}
  \pgfmathparse{\pgfgetarrowoptions{thickarrow}}%
  \ifpgfmathunitsdeclared%
    \pgfmathparse{\pgfmathresult pt}%
  \else%
    \pgfmathparse{\pgfmathresult*\pgflinewidth}%
  \fi%
  \let\thickness=\pgfmathresult
  \pgfsetlinewidth{1.25 pt}
  \pgfutil@tempdima=0.28pt%
  \advance\pgfutil@tempdima by.3\pgflinewidth%
  \pgfsetlinewidth{0.8\pgflinewidth}
  \pgfsetdash{}{+0pt}
  \pgfsetroundcap
  \pgfsetroundjoin
  \pgfpathmoveto{\pgfqpoint{-3\pgfutil@tempdima}{4\pgfutil@tempdima}}
  \pgfpathcurveto
  {\pgfqpoint{-2.75\pgfutil@tempdima}{2.5\pgfutil@tempdima}}
  {\pgfqpoint{0pt}{0.25\pgfutil@tempdima}}
  {\pgfqpoint{0.75\pgfutil@tempdima}{0pt}}
  \pgfpathcurveto
  {\pgfqpoint{0pt}{-0.25\pgfutil@tempdima}}
  {\pgfqpoint{-2.75\pgfutil@tempdima}{-2.5\pgfutil@tempdima}}
  {\pgfqpoint{-3\pgfutil@tempdima}{-4\pgfutil@tempdima}}
  \pgfusepathqstroke
}
\makeatother

\makeatletter
\pgfdeclareshape{ground}
{
    \savedanchor\centerpoint
    {
        \pgf@x=0pt
        \pgf@y=0pt
    }
    \anchor{center}{\centerpoint}
    \anchorborder{\centerpoint}
    \anchor{north}
    {
        \pgf@x=0pt
        \pgf@y=0.5*0.33*\stateheight
    }
    \anchor{south}
    {
        \pgf@x=0pt
        \pgf@y=0pt
    }
    \saveddimen\overallwidth
    {
        \pgfkeysgetvalue{/pgf/minimum width}{\minwidth}
        \pgf@x=\minimumstatewidth
        \ifdim\pgf@x<\minwidth
            \pgf@x=\minwidth
        \fi
    }
    \backgroundpath
    {
        \begin{pgfonlayer}{foreground}
        \pgfsetstrokecolor{black}
        \pgfsetlinewidth{1.25pt}
        \ifhflip
            \pgftransformyscale{-1}
        \fi
        \pgftransformscale{0.5}
        \pgfpathmoveto{\pgfpoint{-0.5*\overallwidth}{0}}
        \pgfpathlineto{\pgfpoint{0.5*\overallwidth}{0}}
        \pgfpathmoveto{\pgfpoint{-0.33*\overallwidth}{0.33*\stateheight}}
        \pgfpathlineto{\pgfpoint{0.33*\overallwidth}{0.33*\stateheight}}
        \pgfpathmoveto{\pgfpoint{-0.16*\overallwidth}{0.66*\stateheight}}
        \pgfpathlineto{\pgfpoint{0.16*\overallwidth}{0.66*\stateheight}}
        \pgfpathmoveto{\pgfpoint{-0.02*\overallwidth}{\stateheight}}
        \pgfpathlineto{\pgfpoint{0.02*\overallwidth}{\stateheight}}
        \pgfusepath{stroke}
        \end{pgfonlayer}
    }
}
\tikzset{forward arrow style/.style={every to/.style, decoration={
    markings,
    mark=at position 0.5 with \arrow{thickarrow}},
    postaction=decorate}}
\tikzset{reverse arrow style/.style={every to/.style, decoration={
    markings,
    mark=at position 0.5 with \arrow{reversethickarrow}},
    postaction=decorate}}
\pgfdeclareshape{morphism}
{
    \savedanchor\centerpoint
    {
        \pgf@x=0pt
        \pgf@y=0pt
    }
    \anchor{center}{\centerpoint}
    \anchorborder{\centerpoint}
    \saveddimen\savedlengthnw
    {
        \pgfkeysgetvalue{/tikz/keylengthnw}{\len}
        \pgf@x=\len
    }
    \saveddimen\savedlengthn
    {
        \pgfkeysgetvalue{/tikz/keylengthn}{\len}
        \pgf@x=\len
    }
    \saveddimen\savedlengthne
    {
        \pgfkeysgetvalue{/tikz/keylengthne}{\len}
        \pgf@x=\len
    }
    \saveddimen\savedlengthsw
    {
        \pgfkeysgetvalue{/tikz/keylengthsw}{\len}
        \pgf@x=\len
    }
    \saveddimen\savedlengths
    {
        \pgfkeysgetvalue{/tikz/keylengths}{\len}
        \pgf@x=\len
    }
    \saveddimen\savedlengthse
    {
        \pgfkeysgetvalue{/tikz/keylengthse}{\len}
        \pgf@x=\len
    }
    \saveddimen\overallwidth
    {
        \pgfkeysgetvalue{/tikz/width}{\minwidth}
        \pgf@x=\wd\pgfnodeparttextbox
        \ifdim\pgf@x<\minwidth
            \pgf@x=\minwidth
        \fi
    }
    \savedanchor{\upperrightcorner}
    {
        \pgf@y=.5\ht\pgfnodeparttextbox
        \advance\pgf@y by -.5\dp\pgfnodeparttextbox
        \pgf@x=.5\wd\pgfnodeparttextbox
    }
    \anchor{north}
    {
        \pgf@x=0pt
        \pgf@y=0.5\morphismheight
    }
    \anchor{north east}
    {
        \pgf@x=\overallwidth
        \multiply \pgf@x by 2
        \divide \pgf@x by 5
        \pgf@y=0.5\morphismheight
    }
    \anchor{east}
    {
        \pgf@x=\overallwidth
        \divide \pgf@x by 2
        \advance \pgf@x by 5pt
        \pgf@y=0pt
    }
    \anchor{west}
    {
        \pgf@x=-\overallwidth
        \divide \pgf@x by 2
        \advance \pgf@x by -5pt
        \pgf@y=0pt
    }
    \anchor{north west}
    {
        \pgf@x=-\overallwidth
        \multiply \pgf@x by 2
        \divide \pgf@x by 5
        \pgf@y=0.5\morphismheight
    }
    \anchor{connect nw}
    {
        \pgf@x=-\overallwidth
        \multiply \pgf@x by 2
        \divide \pgf@x by 5
        \pgf@y=0.5\morphismheight
        \advance\pgf@y by \savedlengthnw
    }
    \anchor{connect ne}
    {
        \pgf@x=\overallwidth
        \multiply \pgf@x by 2
        \divide \pgf@x by 5
        \pgf@y=0.5\morphismheight
        \advance\pgf@y by \savedlengthne
    }
    \anchor{connect sw}
    {
        \pgf@x=-\overallwidth
        \multiply \pgf@x by 2
        \divide \pgf@x by 5
        \pgf@y=-0.5\morphismheight
        \advance\pgf@y by -\savedlengthsw
    }
    \anchor{connect se}
    {
        \pgf@x=\overallwidth
        \multiply \pgf@x by 2
        \divide \pgf@x by 5
        \pgf@y=-0.5\morphismheight
        \advance\pgf@y by -\savedlengthse
    }
    \anchor{connect n}
    {
        \pgf@x=0pt
        \pgf@y=0.5\morphismheight
        \advance\pgf@y by \savedlengthn
    }
    \anchor{connect s}
    {
        \pgf@x=0pt
        \pgf@y=-0.5\morphismheight
        \advance\pgf@y by -\savedlengths
    }
    \anchor{south east}
    {
        \pgf@x=\overallwidth
        \multiply \pgf@x by 2
        \divide \pgf@x by 5
        \pgf@y=-0.5\morphismheight
    }
    \anchor{south west}
    {
        \pgf@x=-\overallwidth
        \multiply \pgf@x by 2
        \divide \pgf@x by 5
        \pgf@y=-0.5\morphismheight
    }
    \anchor{south}
    {
        \pgf@x=0pt
        \pgf@y=-0.5\morphismheight
    }
    \anchor{text}
    {
        \upperrightcorner
        \pgf@x=-\pgf@x
        \pgf@y=-\pgf@y
    }
    \backgroundpath
    {
        \pgfsetstrokecolor{black}
        \pgfsetlinewidth{\thickness}
        \begin{scope}
                \ifhflip
                    \pgftransformyscale{-1}
                \fi
                \ifvflip
                    \pgftransformxscale{-1}
                \fi
                \ifhvflip
                    \pgftransformxscale{-1}
                    \pgftransformyscale{-1}
                \fi
                \pgfpathmoveto{\pgfpoint
                    {-0.5*\overallwidth-5pt}
                    {0.5*\morphismheight}}
                \pgfpathlineto{\pgfpoint
                    {0.5*\overallwidth+5pt}
                    {0.5*\morphismheight}}
                \ifwedge
                    \pgfpathlineto{\pgfpoint
                        {0.5*\overallwidth + 15pt}
                        {-0.5*\morphismheight}}
                \else
                    \pgfpathlineto{\pgfpoint
                        {0.5*\overallwidth + 5pt}
                        {-0.5*\morphismheight}}
                \fi
                \pgfpathlineto{\pgfpoint
                    {-0.5*\overallwidth-5pt}
                    {-0.5*\morphismheight}}
                \pgfpathclose
                \pgfusepath{stroke}
        \end{scope}
        \ifconnectnw
            \pgfpathmoveto{\pgfpoint
                {-0.4*\overallwidth}
                {0.5*\morphismheight}}
            \pgfpathlineto{\pgfpoint
                {-0.4*\overallwidth}
                {0.5*\morphismheight+\savedlengthnw}}
            \pgfusepath{stroke}
        \fi
        \ifconnectne
            \pgfpathmoveto{\pgfpoint
                {0.4*\overallwidth}
                {0.5*\morphismheight}}
            \pgfpathlineto{\pgfpoint
                {0.4*\overallwidth}
                {0.5*\morphismheight+\savedlengthne}}
            \pgfusepath{stroke}
        \fi
        \ifconnectsw
            \pgfpathmoveto{\pgfpoint
                {-0.4*\overallwidth}
                {-0.5*\morphismheight}}
            \pgfpathlineto{\pgfpoint
                {-0.4*\overallwidth}
                {-0.5*\morphismheight-\savedlengthsw}}
            \pgfusepath{stroke}
        \fi
        \ifconnectse
            \pgfpathmoveto{\pgfpoint
                {0.4*\overallwidth}
                {-0.5*\morphismheight}}
            \pgfpathlineto{\pgfpoint
                {0.4*\overallwidth}
                {-0.5*\morphismheight-\savedlengthse}}
            \pgfusepath{stroke}
        \fi
        \ifconnectn
            \pgfpathmoveto{\pgfpoint
                {0pt}
                {0.5*\morphismheight}}
            \pgfpathlineto{\pgfpoint
                {0pt}
                {0.5*\morphismheight+\savedlengthn}}
            \pgfusepath{stroke}
        \fi
        \ifconnects
            \pgfpathmoveto{\pgfpoint
                {0pt}
                {-0.5*\morphismheight}}
            \pgfpathlineto{\pgfpoint
                {0pt}
                {-0.5*\morphismheight-\savedlengths}}
            \pgfusepath{stroke}
        \fi
        \ifconnectnwf
            \draw [forward arrow style] (-0.4*\overallwidth,0.5*\morphismheight)
                to (-0.4*\overallwidth,0.5*\morphismheight+\savedlengthnw);
        \fi
        \ifconnectnef
            \draw [forward arrow style] (0.4*\overallwidth,0.5*\morphismheight)
                to (0.4*\overallwidth,0.5*\morphismheight+\savedlengthne);
        \fi
        \ifconnectswf
            \draw [forward arrow style] (-0.4*\overallwidth,-0.5*\morphismheight-\savedlengthsw)
                to (-0.4*\overallwidth,-0.5*\morphismheight);
        \fi
        \ifconnectsef
            \draw [forward arrow style] (0.4*\overallwidth,-0.5*\morphismheight-\savedlengthse)
                to (0.4*\overallwidth,-0.5*\morphismheight);
        \fi
        \ifconnectnf
            \draw [forward arrow style] (0,0.5*\morphismheight)
                to (0,0.5*\morphismheight+\savedlengthn);
        \fi
        \ifconnectsf
            \draw [forward arrow style] (0,-0.5*\morphismheight-\savedlengths)
                to (0,-0.5*\morphismheight);
        \fi
        \ifconnectnwr
            \draw [reverse arrow style] (-0.4*\overallwidth,0.5*\morphismheight)
                to (-0.4*\overallwidth,0.5*\morphismheight+\savedlengthnw);
        \fi
        \ifconnectner
            \draw [reverse arrow style] (0.4*\overallwidth,0.5*\morphismheight)
                to (0.4*\overallwidth,0.5*\morphismheight+\savedlengthne);
        \fi
        \ifconnectswr
            \draw [reverse arrow style] (-0.4*\overallwidth,-0.5*\morphismheight-\savedlengthsw)
                to (-0.4*\overallwidth,-0.5*\morphismheight);
        \fi
        \ifconnectser
            \draw [reverse arrow style] (0.4*\overallwidth,-0.5*\morphismheight-\savedlengthse)
                to (0.4*\overallwidth,-0.5*\morphismheight);
        \fi
        \ifconnectnr
            \draw [reverse arrow style] (0,0.5*\morphismheight)
                to (0,0.5*\morphismheight+\savedlengthn);
        \fi
        \ifconnectsr
            \draw [reverse arrow style] (0,-0.5*\morphismheight-\savedlengths)
                to (0,-0.5*\morphismheight);
        \fi
    }
}
\pgfdeclareshape{swish right}
{
    \savedanchor\centerpoint
    {
        \pgf@x=0pt
        \pgf@y=0pt
    }
    \anchor{center}{\centerpoint}
    \anchorborder{\centerpoint}
    \anchor{north}
    {
        \pgf@x=\minimummorphismwidth
        \divide\pgf@x by 5
        \pgf@y=\morphismheight
        \divide\pgf@y by 2
        \advance\pgf@y by \connectheight
    }
    \anchor{south}
    {
        \pgf@x=-\minimummorphismwidth
        \divide\pgf@x by 5
        \pgf@y=-\morphismheight
        \divide\pgf@y by 2
        \advance\pgf@y by -\connectheight
    }
    \backgroundpath
    {
        \pgfsetstrokecolor{black}
        \pgfsetlinewidth{\thickness}
        \pgfpathmoveto{\pgfpoint
            {-0.2*\minimummorphismwidth}
            {-0.5*\morphismheight-\connectheight}}
        \pgfpathcurveto
            {\pgfpoint{-0.2*\minimummorphismwidth}{0pt}}
            {\pgfpoint{0.2*\minimummorphismwidth}{0pt}}
            {\pgfpoint
                {0.2*\minimummorphismwidth}
                {0.5*\morphismheight+\connectheight}}
        \pgfusepath{stroke}
    }
}
\pgfdeclareshape{swish left}
{
    \savedanchor\centerpoint
    {
        \pgf@x=0pt
        \pgf@y=0pt
    }
    \anchor{center}{\centerpoint}
    \anchorborder{\centerpoint}
    \anchor{north}
    {
        \pgf@x=-\minimummorphismwidth
        \divide\pgf@x by 5
        \pgf@y=\morphismheight
        \divide\pgf@y by 2
        \advance\pgf@y by \connectheight
    }
    \anchor{south}
    {
        \pgf@x=\minimummorphismwidth
        \divide\pgf@x by 5
        \pgf@y=-\morphismheight
        \divide\pgf@y by 2
        \advance\pgf@y by -\connectheight
    }
    \backgroundpath
    {
        \pgfsetstrokecolor{black}
        \pgfsetlinewidth{\thickness}
        \pgfpathmoveto{\pgfpoint
            {0.2*\minimummorphismwidth}
            {-0.5*\morphismheight-\connectheight}}
        \pgfpathcurveto
            {\pgfpoint{0.2*\minimummorphismwidth}{0pt}}
            {\pgfpoint{-0.2*\minimummorphismwidth}{0pt}}
            {\pgfpoint
                {-0.2*\minimummorphismwidth}
                {0.5*\morphismheight+\connectheight}}
        \pgfusepath{stroke}
    }
}
\pgfdeclareshape{state}
{
    \savedanchor\centerpoint
    {
        \pgf@x=0pt
        \pgf@y=0pt
    }
    \anchor{center}{\centerpoint}
    \anchorborder{\centerpoint}
    \saveddimen\overallwidth
    {
        \pgf@x=3\wd\pgfnodeparttextbox
        \ifdim\pgf@x<\minimumstatewidth
            \pgf@x=\minimumstatewidth
        \fi
    }
    \savedanchor{\upperrightcorner}
    {
        \pgf@x=.5\wd\pgfnodeparttextbox
        \pgf@y=.5\ht\pgfnodeparttextbox
        \advance\pgf@y by -.5\dp\pgfnodeparttextbox
    }
    \anchor{A}
    {
        \pgf@x=-\overallwidth
        \divide\pgf@x by 4
        \pgf@y=0pt
    }
    \anchor{B}
    {
        \pgf@x=\overallwidth
        \divide\pgf@x by 4
        \pgf@y=0pt
    }
    \anchor{text}
    {
        \upperrightcorner
        \pgf@x=-\pgf@x
        \ifhflip
            \pgf@y=-\pgf@y
            \advance\pgf@y by 0.4\stateheight
        \else
            \pgf@y=-\pgf@y
            \advance\pgf@y by -0.4\stateheight
        \fi
    }
    \backgroundpath
    {
        \pgfsetstrokecolor{black}
        \pgfsetlinewidth{\thickness}
        \pgfpathmoveto{\pgfpoint{-0.5*\overallwidth}{0}}
        \pgfpathlineto{\pgfpoint{0.5*\overallwidth}{0}}
        \ifhflip
            \pgfpathlineto{\pgfpoint{0}{\stateheight}}
        \else
            \pgfpathlineto{\pgfpoint{0}{-\stateheight}}
        \fi
        \pgfpathclose
        \ifblack
            \pgfsetfillcolor{black!50}
            \pgfusepath{fill,stroke}
        \else
            \pgfusepath{stroke}
        \fi
    }
}
\makeatother

\newcommand{\tinycomult}[1][dot]{
\smash{\raisebox{-1.8pt}{\hspace{-5pt}\ensuremath{\begin{pic}[scale=0.34,string]
    \node (0) at (0,0) {};
    \node[#1, inner sep=1.5pt] (1) at (0,0.55) {};
    \node (2) at (-0.5,1) {};
    \node (3) at (0.5,1) {};
    \draw (0.center) to (1.center);
    \draw (1.center) to [out=left, in=down, out looseness=1.5] (2.center);
    \draw (1.center) to [out=right, in=down, out looseness=1.5] (3.center);
\end{pic}
}\hspace{-3pt}}}}
\newcommand{\tinycounit}[1][dot]{
\smash{\raisebox{-1.8pt}{\ensuremath{\hspace{-3pt}\begin{pic}[scale=0.34,string]
        \node (0) at (0,0) {};
        \node (1) at (0,1) {};
        \node[#1, inner sep=1.5pt] (d) at (0,0.55) {};
        \draw (0.center) to (d.center);
    \end{pic}
    \hspace{-1pt}}}}}
\newcommand{\tinymult}[1][dot]{
\smash{{\hspace{-5pt}\ensuremath{\begin{aligned}\begin{tikzpicture}[scale=0.36,thick,yscale=-1]
    \node (0) at (0,0) {};
    \node (2) at (-0.5,1) {};
    \node (3) at (0.5,1) {};
    \draw (0.center) to (0,0.55);
    \draw (0,0.55) to [out=left, in=down, out looseness=1.5] (2.center);
    \draw (0,0.55) to [out=right, in=down, out looseness=1.5] (3.center);
    \node[#1, inner sep=1.5pt] (1) at (0,0.55) {};
\end{tikzpicture}\end{aligned}
}\hspace{-3pt}}}}
\newcommand{\tinyunit}[1][dot]{
\smash{{\ensuremath{\hspace{-3pt}\begin{aligned}\begin{tikzpicture}[scale=0.4,thick,yscale=-1]
        \node (0) at (0,0) {};
        \node (1) at (0,1) {};
        \draw (0.center) to (0,0.55);
        \node[#1, inner sep=1.5pt] (d) at (0,0.55) {};
    \end{tikzpicture}\end{aligned}
    \hspace{-1pt}}}}}

\newcommand{\tinyhandle}[1][dot]{\raisebox{-2pt}{\ensuremath{\hspace{-3pt}\begin{pic}[scale=0.4,string]
        \node (0) at (0,0) {};
        \node[dot, inner sep=1.0pt] (1) at (0,0.3) {};
        \node[dot, inner sep=1.0pt] (2) at (0,0.7) {};
        \node (3) at (0,1) {};
        \draw (0.center) to (1.center);
        \draw (2.center) to (3.center);
        \draw[in=180, out=180, looseness=2] (1.center) to (2.center);
        \draw[in=0, out=0, looseness=2] (1.center) to (2.center);
\end{pic}\hspace{-1pt}}}}

\newcommand\cat[1]{\ensuremath{\mathbf{#1}}}
\renewcommand\dag{\ensuremath{\dagger}}

\newcommand\sxrightarrow[1]{\xrightarrow{\smash{#1}}}

\newcommand\sxto[1]{\sxrightarrow{#1}}

\newcommand\ket[1]{\ensuremath{| #1 \rangle}}
\newcommand\bra[1]{\ensuremath{\langle #1 |}}

\newcommand\ud{\ensuremath{\mathrm{d}}}
\usepackage{color}

\def\swangle{-145}
\def\seangle{-35}
\def\nwangle{145}
\def\neangle{35}

\def\newyscale{0.8}
\def\newxscale{0.75}

\ignore{
\usepackage{mathtools}
\tikzstyle{dot}=[inner sep=0.7mm,minimum width=0pt,minimum
height=0pt,fill=black,draw=black,shape=circle]
\tikzstyle{black dot}=[dot,fill=black]
\tikzstyle{white dot}=[dot,fill=white]
\tikzstyle{gray dot}=[dot,fill=gray!40!white]
}